%% file: main.tex
\documentclass[sigconf,nonacm]{acmart}
\usepackage{modules}
\usepackage{hyperref}
\usepackage{cleveref}
\usepackage{notations}

\usepackage{setspace}

\newenvironment{lemmaclone}[1]{\noindent
{Lemma~\ref{#1}} (restated).\em}{\par}

\AtBeginDocument{%
  \providecommand\BibTeX{{%
    \normalfont B\kern-0.5em{\scshape i\kern-0.25em b}\kern-0.8em\TeX}}}

\begin{document}

\title{Economically Viable Randomness}

\author{David Yakira}
\affiliation{%
  \institution{Orbs, Technion}}
\email{david@orbs.com}

\author{Avi Asayag}
\affiliation{%
  \institution{Technion}}
\email{avi.asayag@campus.technion.ac.il}

\author{Ido Grayevsky}
\affiliation{%
  \institution{University of Oxford}}
\email{ido.grayevsky@maths.ox.ac.uk}

\author{Idit Keidar}
\affiliation{%
  \institution{Technion}}
\email{idish@ee.technion.ac.il}



\begin{abstract}
    \input{abstract.tex}

\end{abstract}


\maketitle

\section{Introduction}
\label{sec:intro}
\input{intro.tex}

\section{Model}
\label{sec:model}
\input{abstraction_model.tex}

\section{Problem Definition: EVR}
\label{sec:EVR_primitive}
\input{EVR_ideal_sources_new}

\section{An EVR Source as a Service in \texorpdfstring{$\Z$}{B}}
\label{sec:EVR_realization}
\input{EVR_source_service}

\section{Game-Theoretic Analysis}
\label{sec:analysis}
\input{analysis.tex}

\section{Multi-Shot EVR}
\label{sec:multi_shot_EVRS}
\input{multi-shot_ideal_source}

\section{Practical Considerations}
\label{sec:implementation}
\input{prototype}

\section{Related Work}
\label{sec:related}
\input{related_work_2}

\section{Conclusion}
\label{sec:conclusion}
\input{conclusion}

\bibliographystyle{ACM-Reference-Format}
\bibliography{bibliography}

\appendix
\section{Formal Proofs}
\label{app:omitted-proofs}
\input{proofs/Lemma_1}

\input{proofs/Lemma_2}

\input{proofs/Lemma_3}

\end{document}

%% file: abstract.tex
We study the problem of providing blockchain applications with \emph{economically viable randomness} (EVR), namely, randomness that has significant economic consequences. Applications of EVR include blockchain-based lotteries and gambling. An EVR source guarantees (i) \emph{secrecy}, assuring that the random bits are kept secret until some predefined condition indicates that they are safe to reveal (e.g., the lottery's ticket sale closes), and (ii) \emph{robustness}, guaranteeing that the random bits are published once the condition holds.
We formalize the EVR problem and solve it on top of an Ethereum-like blockchain abstraction, which supports smart contracts and a transferable native coin. Randomness is generated via a distributed open commit-reveal scheme by game-theoretic agents who strive to maximize their coin holdings. Note that in an economic setting, such agents might profit from breaking secrecy or robustness, and may engage in side agreements (via smart contracts) to this end. Our solution creates an incentive structure that counters such attacks. We prove that following the protocol gives rise to a stable state, called Coalition-Proof Nash Equilibrium, from which no coalition comprised of a subset of the players can agree to deviate. In this stable state, robustness and secrecy are satisfied. Finally, we implement our EVR source over Ethereum.

%% file: intro.tex
\subsection{Motivation and Goal}
In today's digital gambling industry, bit strings presumed to be pseudo-random are generated behind closed doors by centralized services such as online casinos and lotteries. These random bits have significant economic implications in allocating valuable prizes to winners. Yet, these services are not subject to public audit and their integrity thus must be blindly trusted by their users (the gamblers).

Decentralized smart contract platforms like Ethereum~\cite{ethereum,EthereumButerin2013whitepaper} are creating a paradigm shift in the way digital services are built. They facilitate \emph{decentralized applications}, \emph{dApps} for short, which are deployed as smart contracts~\cite{szabo1997formalizing} and are guaranteed to follow well-specified code. Due to their significant economic implications and susceptibility to manipulation, gambling games could be a perfect match for this new paradigm. Indeed, ever since it was launched, Ethereum has experienced a proliferation of gaming and gambling dApps that require randomness~\cite{dappradar}. However, the randomness generation processes that these dApps adopt today lack clear guarantees and there is no way to evaluate their resilience to selfish economic agents who strive to maximize their coin holdings.

This paper addresses the need for trustworthy randomness. We provide for the first time a randomness source for economic game-theoretic settings, which dApps can trigger on-demand. We say that such a source provides \emph{Economically Viable Randomness (EVR)}. 

Yet, realizing an EVR source involves a number of intrinsic challenges. First, if the random bits have high stakes (e.g., in allocating a large jackpot) there is a strong economic pressure to tamper with them. It is challenging to counter this pressure without requiring a high collateral.
Second, perhaps paradoxically, the same decentralized and open nature that blockchain platforms offer also renders them inherently deterministic and therefore incapable of generating random numbers. Any attempt to enhance dApps with randomness therefore has to incorporate an off-chain component and use the blockchain to verify the authenticity of the random value it provides. To be trustworthy, such an off-chain component must be decentralized and combine randomness from many independent sources. This creates the challenge of scalability, which is exacerbated in blockchain-based solutions due to their fairly low capacity and high costs. In order to achieve scalability, one needs to steer most of the heavy lifting off-chain.


In what follows, we describe our approach to providing EVR. 



\subsection{Contributions}
\paragraph{Blockchain abstraction}
In order to conceptualize our solution independently of any particular blockchain technology, we define in \Cref{sec:model} an abstraction, denoted $\Z$, capturing the pertinent aspects of a blockchain offering smart contracts and a native coin. The coin sets the ground for a game-theoretic model, where agents strive to maximize their coin holdings in $\Z$. Smart contracts allow us to implement trusted services. 
We shall use them to realize an escrow service that handles our game-theoretic agents' coins in a way that incentivizes them to behave in accordance with the system's desired outcome. On the other hand, smart contracts can be used by our agents to facilitate trustless agreements among them in attempts to break our scheme. The blockchain abstraction also concretizes the notion of publishing information -- anything written on $\Z$ is irrevocably publicly available.

\paragraph{An EVR source}
In \Cref{sec:EVR_primitive} we formalize the problem that we set to solve in this work. We begin by defining an ideal randomness source that captures the essence of generating trustworthy randomness via a commit-reveal scheme.
The application for which the randomness is provided specifies a condition (e.g., a time in the future) indicating when it would like the randomness to be revealed. Then, an ideal source satisfies (i) \emph{secrecy} -- the random value does not leak before the condition holds; and (ii) \emph{robustness} -- when the condition is met, the random value is indeed revealed.

An ideal source can be used in a simple lottery dApp as follows: Gamblers buy lottery tickets by sending coins to the lottery's contract and indicating a 256-bit string. During the ticket sale, the randomness must remain completely obscure so that no one can gain an unfair advantage in buying a winning ticket. The condition to reveal the randomness materializes slightly after the ticket sale closes. Then, the lottery's account obtains a random 256-bit string with which it can infer the winning tickets.         

While an ideal source satisfies secrecy and robustness under all circumstances, an EVR source emulates an ideal one as long as the application abides with explicit economic bounds on the value of the randomness, defined by the source. For example, the application must maintain its total payout (to gamblers) under some bound specified by the EVR source. In the lottery example, this can be achieved by limiting the total amount raised from ticket sales.


\paragraph{Our solution}
In \Cref{sec:EVR_realization} we describe our realization of an EVR source. Our EVR source is distributed, with open participation: any blockchain account owner who wishes to partake in producing the randomness needs to \emph{register} within a designated escrow service, realized as a smart contract on $\Z$, while depositing a small collateral of $1$ coin. The registered accounts' owners are the \emph{players} that realize the EVR source. We use an escrow-mediated distributed key generation (DKG) protocol~\cite{rational_threshold_cryptosystems,EthDKG_austrians} in order to achieve robustness for high-stake lotteries despite using small deposits, while doing most of the computations off-chain with lightweight on-chain verification. 
Our solution guarantees secrecy through a novel \emph{informing} mechanism that allows anyone who knows the random value during the period when it should remain obscure to report it on the escrow contract for a substantial compensation, funded by the players' deposits. 

\paragraph{Game-theoretic analysis}
In \Cref{sec:analysis} we analyze our solution in a game-theoretic setting. Our main result is compelling -- we show that players are incentivized to follow the default strategy (namely, follow the protocol honestly) ensuring both secrecy and robustness. Particularly, in the default strategy, players who can inform do so, and our analysis shows that under well-specified conditions, informing is profitable and thus effective in deterring collusion in attempt to break secrecy.

More formally, we show that following the default strategy gives rise to a \emph{Coalition-Proof Nash equilibrium (CPNE)}~\cite{CPNE}, a strong equilibrium concept in game-theory that withstands any self-enforcing deviation by a coalition (subset of players) of any size. Put differently, if a strategy is a CPNE, there is no coalition that can deviate from it to the benefit of all of its members while no sub-coalition can further deviate and gain more, making the original deviation unstable in the first place. Importantly, our result holds even if players may engage in side agreements that are enforced via smart contracts. 

\paragraph{Practical considerations}
While most of the paper discusses single-shot EVR, where a single random value is generated and published, in \Cref{sec:multi_shot_EVRS} we describe the multi-shot version, where a single commitment produces multiple random values revealed in succession as the circumstances ripen to reveal them. This generalization is needed by real-life applications such as card and die games. Then, in \Cref{sec:implementation} we present a proof-of-concept multi-shot EVR implementation, where the escrow service is realized as a smart contract on Ethereum. We address real-world issues that arise due to Ethereum-specific limitations. Our solution scales to hundreds of participating accounts, within Ethereum's block gas limit.

In \Cref{sec:related} we compare our approach to distributed randomness solutions in the literature. Our work is unique in its underlying economic model that assumes that all  players are selfish. For this model, we manage to design a general service that requires constant deposits (for a jackpot that scales linearly with the number of players) and that can scale in the number of players that contribute to the randomness generation process. Finally, \Cref{sec:conclusion} concludes the paper.


%% file: abstraction_model.tex
In this work we design a service based on today's 
blockchain technologies that embody smart contracts (e.g., Ethereum).
In order to formalize our service in a general way, we define in \Cref{smart_contract_abstraction} an abstraction, $\Z$, capturing the pertinent features of the underlying blockchain. Our abstraction is based on Ethereum, but intentionally abstracts away Ethereum-specific technicalities like mining and generation of coins (ether), transaction fees and gas, the mempool, the P2P network, etc. We revisit these in \Cref{sec:implementation}, where we discuss our implementation over Ethereum.
In \Cref{sec:service}, we define the notion of a \emph{service} in $\Z$.

\subsection{Smart Contract Platform} \label{smart_contract_abstraction}
$\Z$ is a smart contract platform abstraction that is accessed by a collection of \emph{users}. It consists of four components:
\begin{enumerate}
    \item A global clock.
    \item \emph{Accts} -- a key-value map from account identifiers to account tuples of the form $(S, C, s_0)$, where: (1) $S$ is a (potentially infinite) set of states represented as a collection of state variables, one of which is the account's balance $b\in \mathbb{N}_0$. (2) $C$ is the code defining the account's logic. The code is organized in deterministic functions that manipulate the account's state (i.e., define state transitions).
    Functions are triggered either by \emph{transactions} that users generate, or by other functions that had been previously triggered (cross-account function triggering is possible). Functions can access the global clock. (3) $s_0\in S$ is the account's initial state, which includes $b_0$, the initial balance.
    \item \emph{Log} -- an append-only list of user-issued transactions. Every transaction begins by invoking a function on some account and that function may invoke nested calls to other accounts' functions. 
    \item \emph{State Machine} -- a deterministic machine that processes any sequence of transactions according to the appropriate account codes.
\end{enumerate}

In $\Z$'s \emph{initial state}, all accounts in Accts are in their initial states.
At any point in time, $\Z$'s \emph{current state} is
composed of
the states of all accounts after the State Machine processes  the sequence of transactions that are in  $\Z$'s Log at that time, starting from $\Z$'s initial state. The State Machine maintains a \emph{conservation law} of Accts' balances -- the sum of balances in all accounts is invariant. This way, balances represent \emph{coins}.

Users \emph{write} to $\Z$ by issuing transactions and appending them to the Log. They can also \emph{read} $\Z$'s current state. Writing and reading happen instantaneously. (We implicitly assume that transactions execute a bounded number of steps. Ethereum ensures this in practice by defining a block gas limit.) 

Computations that are invoked by transactions and are processed by $\Z$'s State Machine are said to occur \emph{on-chain}. Conversely, users can execute calculations \emph{off-chain}, i.e., on private probabilistic computationally-bounded machines.

A \emph{verifiable condition} in $\Z$ is a predicate evaluated against $\Z$'s current state. A verifiable condition that evaluates to \true{} and remains \true{} in all possible future states of $\Z$ is said to have \emph{matured}. Such conditions may be time-dependant or depend on specific transactions having been appended to $\Z$'s Log.

We distinguish between two types of accounts in $\Z$. Using Ethereum terminology, an \emph{externally-owned account} (EOA) is an account whose state consists only of a balance.
An EOA's identifier is a public key, $pk$, such that in order to spend the account's coins, the account's code validates a specific signature against $pk$. Thus, an EOA is ``owned'' by the user who has access to the secret key that corresponds to its identifier.

A \emph{smart contract} is an account that no user owns. Thus, a smart contract's code indicates the terms that allow spending its coins. In their initial states, smart contracts have zero balance, and subsequently they can receive coins from other accounts so each coin in a smart contract can be traced to its EOA of origin. 

A specific type of smart contract that we are interested in is the \emph{escrow account}. Such an account holds coins on behalf of other accounts (EOAs and smart contracts), and permits them to withdraw these coins according to some predetermined conditions, encoded in the escrow's code. For example, a chess-betting escrow account takes deposits from two EOAs (say $1$ coin each), oversees a game of chess between them, where each EOA, in turn, makes a move by submitting a transaction (the rules of chess are encoded in the escrow's code), and finally when one of the EOAs wins, it pays the winner the $2$ coins it holds.

For simplicity, $\Z$ abstracts away the process of generating new EOAs and deploying new smart contracts. Thus, we assume that all pertinent EOAs and smart contracts exist in $\Z$'s initial state.

Every transaction is associated with an issuing EOA (that belongs to user who issued the transaction). (In Ethereum, this is the EOA that pays the transaction's gas.) In case multiple transactions are appended to $\Z$'s Log at the same time, they are ordered by the EOA identifiers that issue them. (This simplification deliberately abstracts away miners' freedom to order transactions as they please, and masks issues like front-running and chain reorgs, which we revisit in \Cref{sec:implementation}.)

\subsection{A Service as a Smart Contract} \label{sec:service}
A \emph{service} is an escrow account in $\Z$ that performs some task for a third-party, represented by another smart contract on $\Z$. For example, in this work, we build a service that provides randomness to some gambling dApp. A service begins in a \emph{registration} phase, during which users who wish to partake in the service register by depositing $1$ coin into the service's account. A verifiable condition dictates when registration closes. 


In the context of this work we consider services with permission-less registration -- any user who wishes to do so can register. Moreover, a single user may register multiple times (i.e., deposit multiple coins in the service's account). When a coin is deposited by a smart contract, we attribute it to the owner of the EOA from which it originated. We refer to the users who deposit coins into a service as \emph{players}. By the end of the registration phase, the set of players is determined and fixed. Denote the registered players as $[N] = \{ 1,\dots, N \}$ and the number of coins that player $i\in [N]$ deposited as $a_i\in \mathbb{N}$. Denote $n \triangleq \sum_{i=1}^N a_i$, then the service's balance at the end of registration is $n$ coins. 

We consider a game-theoretic setting, where players follow whatever strategy maximizes their 
final coin balances in $\Z$. 
We implicitly assume that the service (or the entity that utilizes the service) offers players some potential profit, so as to incentivize them to register. Indeed, paying dividends to service providers is common in Proof-of-Stake protocols (e.g., Cosmos~\cite{CosmosWhitePaper}, Tezos~\cite{TezosWhitepaper}, Orbs~\cite{helix}) among others (e.g., Augur~\cite{AugurWhitepaper}, Truebit~\cite{TrueBit}, Livepeer~\cite{LivepeerWhitepaper}). For the sake of this work, we simply assume that enough players register.  

We further assume that a player $i$ has $e_i \ge 0$ \emph{external coins}, namely, coins that $i$ owns independently of the ones she deposited in the service's account. Additionally, players have access to private communication channels among themselves. 

Since players do not trust each other, the only agreements they can engage in are self-enforcing ones, which are, loosely speaking, agreements in which all sides are better off following the agreement. The external coins 
facilitate such trustless agreements. As an illustration, assume player $i$ earns $5$ coins more by taking action $A$ than by taking action $B$. She can deposit $8$ external coins in a smart contract that would give her her coins back only if she took action $B$. In this way, she distorts her original payoff function, and the other players would now trust her to take action $B$. We refer to such smart contracts as \emph{side contracts}. Note that players may also register to a service with side contracts (namely they first transfer coins from EOAs they own to a side contract and then have the contract register by further depositing the coins into the service's account). 

A service is considered sufficiently decentralized if no single registered player is too ``rich'' therein. The \emph{decentralization assumption} captures this notion quantitatively: 
\begin{equation}
\forall i \in [N], e_i + a_i \le \nicefrac{n}{3}.
\label{eqn:high-participation}
\end{equation}
We note that the decentralization assumption is analogous (to some extent) to the requirement that no single miner obtains too much of the hash power in Nakamoto consensus~\cite{Bitcoin}.


%% file: EVR_ideal_sources_new.tex
Consider an application whose users are economically affected by the outcome of a random process, for example, a lottery. Given its economic implications, the random process is a likely target for manipulation. 
Our goal in this section is to identify and formalize the properties that render a randomness source safe to use in such circumstances. 

As a starting point, we assume that the application is implemented as a smart contract in $\Z$. 
Thus, its deterministic logic is guaranteed to execute as specified. By the same token, as $\Z$'s State Machine is deterministic, randomness generation cannot happen purely on-chain. Whenever the application needs a random value, it triggers the randomness source, which must have an off-chain component where the sampling actually happens. Then, the random value is published on $\Z$ for the application to use.

In many games of chance there is a stage where gamblers make choices not knowing what the random value is going to be, and then the random value is revealed and certain gamblers make a profit while others lose. To adjust this to $\Z$'s terms, we let the application set a verifiable condition in $\Z$, $\mathit{cnd}$, that matures at some point. For the application to work as intended, the randomness must remain completely obscure until $\mathit{cnd}$ matures, and then it needs to be revealed in the clear.

In \Cref{ssec:ideal} we define an ideal on-chain randomness source that can be used by applications in this manner. While an ideal source might be hard to realize, we define in \Cref{ssec:EVR} an economically viable randomness source, which defines restrictions on the economic value of the randomness. An application that adheres to these restrictions can use an EVR source instead of an ideal one.


\subsection{Ideal Source}
\label{ssec:ideal}
A \emph{single-shot randomness source} consists of a pair of 
protocols,   \ascii{commit}
and \ascii{reveal}, and a smart contract  $\E$ in $\Z$, as follows:
\begin{itemize}
    \item 
    A (successful) run of \ascii{commit} produces a commitment, $X$, to a random sample $x$ and publishes $X$ on $\E$. Note that a \ascii{commit} run must have an off-chain component. 
    
    \item 
    Following a \ascii{commit} run, a (successful) run of \ascii{reveal} produces the value committed to, $x$, and publishes it on $\E$. The \ascii{commit} and \ascii{reveal} protocols are tied together through a well-known Boolean verification function $\ascii{ver}(x,X)$. For every $X$ and for every $x_1 \ne x_2$ either $\ascii{ver}(x_1,X)=\false{}$ or $\ascii{ver}(x_2,X)=\false{}$.
    
    \item $\E$ 
    exposes the following API:
    \begin{itemize}
        \item A function $\ascii{comTrigger}(\mathit{cnd})$ called by the application that triggers a \ascii{commit} run, where $\mathit{cnd}$ is a verifiable condition in $\Z$, determined by the application, that eventually matures. Once $\mathit{cnd}$ matures, a \ascii{reveal} run begins.
    
        \item Two variables, $\mathit{verCom}$ and $\mathit{verRev}$, that indicate whether the \ascii{commit} and \ascii{reveal} runs (resp.) have succeeded. The variables are initiated to $\perp$ and later turn \true{} or \false{}. $\mathit{verRev}$ is updated via a straightforward check with the \ascii{ver} function, whereas $\mathit{verCom}$ is protocol-specific.
    
        \item Two timeout constants, $t_{\text{com}}$ and $t_{\text{rev}}$, that dictate the maximum time that the \ascii{commit} and \ascii{reveal} runs (resp.) can take. In case one of the runs does not complete in a timely manner, $\E$ sets the appropriate variable to \false{}.
    \end{itemize}
\end{itemize}

The Boolean function \ascii{ver} together with $X$ binds the source to a specific random value. So, once $X$ is published on $\E$, $x$ is determined and no other value would be accepted by $\E$ as the random string.

Given the structure of a randomness source, there are two opportunities for users to ``game'' the application. First, obtaining information about the secret before $\mathit{cnd}$ matures might give a user an advantage relative to what the application had intended. We refer to profits made in this manner as \emph{stealing}.
Second, in case the secret is not published in a timely manner, the application will normally have some \emph{fallback} distribution determining how its coins are distributed, e.g., a refund to the gamblers.
This opens up an opportunity for users to profit from preventing the secret from being published, by causing either $\mathit{verCom}$ or $\mathit{verRev}$ to be set to \false{}.

An \emph{ideal single-shot randomness source}  eliminates both of these risks, and is thus safe to use by an application. Formally, it satisfies the following core properties:
\begin{description}
    \item[Non-triviality] If $\E.\ascii{comTrigger}(\mathit{cnd})$ is called, then $\E.\s\mathit{verCom}$\break turns \true{}.
    
    \item[Hiding secrecy] If $\E.\s\mathit{verCom}$ turns \true{}, then as long as $\mathit{cnd}$ does not mature, no user obtains \emph{any} information about $x$. 

    \item[Robustness] If $\E.\s\mathit{verCom}$ turns \true{}, then after $\mathit{cnd}$ matures\break $\E.\s\mathit{verRev}$ also turns \true{}.
    
\end{description}
Hiding secrecy implies that $X$ is a \emph{hiding} cryptographic commitment. This means that no one can infer any of $x$'s bits with probability greater than $\nicefrac{1}{2}$, implying that $x$ is indeed a random bit string.

\begin{remark}
Generally, a randomness source produces a sequence of random values, one after the other. This is useful for applications that proceed in multiple rounds and need a fresh random value for every round (e.g., card and die games). To keep the presentation concise, during most of this work, we discuss a single-shot version of the EVR source. Later, in \Cref{sec:multi_shot_EVRS}, we extend it to multi-shot EVR, in which a single commitment corresponds to a sequence of random values. 
\end{remark}

An ideal source is safe to use because it enables the application to run as intended. However, 
it is hard to realize in an economic environment where users may gain from gaming the system.

\subsection{EVR Source}\label{ssec:EVR}
We are now ready to define a \emph{single-shot EVR source} that can be used in lieu of an ideal one in an economic environment. An EVR source is a randomness source, but it satisfies the core properties only provided that the application satisfies the economic restrictions laid out in \Cref{def:evrs-usage} below. Additionally, an EVR source might satisfy \emph{secrecy} rather than hiding secrecy: 
\begin{description}
    \item[Secrecy] If $\E.\s\mathit{verCom}$ turns \true{}, then as long as $\mathit{cnd}$ does not mature, no user obtains $x$.    
\end{description}
Secrecy implies that it is computationally infeasible to learn $x$ from $X$ (and from participating in the \ascii{commit} run), but it is weaker than hiding secrecy in the sense that sophisticated users might infer some information about $x$. An EVR source that satisfies hiding secrecy is called a \emph{hiding EVR source}. 

To specify the economic restrictions that an EVR source establishes for the application, we quantify the profits that can be gained by causing the EVR source to deviate from the ideal functionality. Formally, \emph{illicit} profit is the quantity a user gains when an EVR source is used on top of the legitimate profit that the user would have gained if an ideal source were used instead. 
The EVR smart contract, $\E$, evaluates the bound on the illicit profit that the EVR source can sustain and publishes it in a variable we denote by $P$. As long as the application respects this bound and the following economic restrictions, an EVR source provides ideal-like randomness.

\begin{definition}[Single-shot EVR correct usage]
\label{def:evrs-usage}
An application correctly uses a (hiding) EVR source if it satisfies the following conditions:
\begin{enumerate}[label=CU\arabic*.,ref=CU\arabic*]
    \item If non-triviality breaks, no user gains illicit profit. \label{def:evrs-usage:commit}
    
    \item \emph{Fallback profit bound.} If robustness breaks, the total illicit profit gained by all users is less than $P$ coins. \label{def:evrs-usage:reveal}
    
    \item \emph{Stealing bound.} 
    If (hiding) secrecy breaks, the total illicit profit gained by all users is less than $P$ coins. \label{def:evrs-usage:before_cnd}
    
    \item If (hiding) secrecy and robustness are preserved, no user gains illicit profit. \label{def:evrs-usage:after_cnd}
\end{enumerate}
\end{definition}
These restrictions are reflected in the utility functions of the game we define in \Cref{sec:game-def} below.

We next exemplify the interplay between the application and the EVR source to illustrate that the correct usage restrictions are feasible. When an application, for instance a lottery, wishes to use an EVR source, it establishes an adequate condition $\mathit{cnd}$ (e.g., a time after ticket sales close) and calls $\E.\ascii{comTrigger}(\mathit{cnd})$. Once the pair $(P,X)$ is published on $\E$, the application has the randomness committed to and can run its logic, e.g., sell lottery tickets to gamblers. Then, when $\mathit{cnd}$ matures, the application needs $x$ (that corresponds to $X$) in order to complete its logic (in our example, determine the winning lottery tickets) and make payouts to winners.

If $X$ is published and then $x$ is not published in a timely manner, then the application uses its fallback rule to determine who gets its coins. 
In contrast, if $X$ is not published, then \ref{def:evrs-usage:commit} dictates that no illicit profit can be made. This implies that the application cannot make any payouts (and particularly no refunds) in this case. Note that in order to be trusted by users, a dApp will want to refund gamblers in case of failure. Thus, to comply with \ref{def:evrs-usage:commit}, the application should start collecting coins from gamblers only after a successful \ascii{commit} run. 


A straightforward way for an application to comply with \ref{def:evrs-usage:reveal} and \ref{def:evrs-usage:before_cnd} is by limiting its total payouts to $P$ coins. That is, $P$ bounds the jackpot that can be played for, which is typically the amount of coins collected from ticket sales.

Note that when an EVR source satisfies secrecy rather than hiding secrecy, to comply with~\ref{def:evrs-usage:after_cnd} the application must make sure that a gambler with partial information about $x$ (before $\mathit{cnd}$ matures) does not gain an advantage relative to gamblers without such information. To ensure this, $x$ must be used by the application carefully. In our example, if say half of $x$'s bits can be learned by a sophisticated player and the lottery's decision rule is based on $\ascii{XOR}(x,\mathit{ticket})$, then that player has a significant advantage over other players. However, if the decision rule uses $\ascii{hash}(x||\mathit{ticket})$ with a cryptographic hash function, then having partial information about $x$ does not help.

Finally, the application might have to worry about other pitfalls that depend on the specific EVR realization used. For instance, if $x$ is shared among a set of players who can combine data they each privately hold to reconstruct the secret, the application might be vulnerable to collusion via a multi-party computation (MPC)~\cite{Yao-2PC,Yao-MPC,GMW-MPC,Ben-or_MPC1,Chaum-MPC,Rabin-MPC,Beaver-MPC,Canetti-MPC1,Canetti-MPC2}. That is, players might cooperate through an MPC protocol that ensures the privacy of their inputs, computes $x$, and uses it to infer a winning ticket, which is the only part of the computation that is outputted (all steps of the computation are kept obscure). In this way, none of the players actually learns $x$, but they still manage to steal application coins.
To circumvent this issue, the application must choose a decision rule that is ``MPC resistant'', namely that is slow to compute via MPC. As of today, MPCs are not practical; for instance, a state-of-the-art SHA-256 7-party computation takes about 20 seconds (over a fast LAN) \cite{Efficient-MPC}, which is 7 orders of magnitude slower than performing the same computation insecurely \cite{hash_gates,sha256_intel_xeon_computer}. 


%% file: EVR_source_service.tex
We now realize an EVR source. 
In \Cref{sec:escrow_DKG} we overview Escrow-DKG, which is a building block in our solution. Then, in \Cref{sec:realization} we present our protocol. In \Cref{sec:design-rationale} we explain the rationale behind our protocol's design and parameter choices.

\subsection{Background -- Escrow-DKG}
\label{sec:escrow_DKG}

Distributed Key Generation (DKG) protocols~\cite{PedersenDKG,Gennaro_DKG_99,Gennaro03_revisiting} for discrete-log based threshold schemes allow a set of $n$ servers to jointly generate a pair of public and secret keys, $(x,X=g^x)$, in such a way that $X$ is output in the clear while $x$ is shared by the $n$ servers via Shamir secret sharing~\cite{SecretSharing}. Unless an adversary compromises more than a specified threshold $t$ out of the $n$ servers, $x$ remains secret and its shares can be subsequently used by the servers to jointly compute $x$. The secret shares can alternatively be used jointly to perform other cryptographic tasks, e.g., decryption and signatures; (this will become useful for the multi-shot EVR in \Cref{sec:multi_shot_EVRS}). We refer to the key generation protocol as \ascii{DKGCommit} and to the subsequent protocol that combines shares in order to reconstruct $x$ as \ascii{DKGReveal}.


Escrow-DKG~\cite{rational_threshold_cryptosystems} is a DKG protocol variant in the Joint-Feldman family~\cite{Gennaro_DKG_99,Gennaro03_revisiting,PedersenDKG}. It differs from the traditional protocols in its underlying model. 
Whereas traditional DKG protocols assume an adversary that corrupts up to $t$ servers, Escrow-DKG assumes an economic model where all players are rational. Additionally, as its name hints, it assumes a trusted escrow service that substitutes and enhances the broadcast channel usually assumed in these protocols. While runs of traditional DKG protocols always succeed, Escrow-DKG might fail. 

An Escrow-DKG run begins with a permission-less registration phase (as defined for a service in \Cref{sec:service}), where players deposit $1$ coin per secret share they stand to obtain. After registration, the set of players is fixed. The players then engage in a run of \ascii{DKGCommit}. The run fails if the escrow detects that one (or more) of the players has deviated from the protocol. In such an event, the escrow declares the failure, returns deposits tied to honest shares, and confiscates deposits tied to misbehaving shares. When there is nothing to gain from failing a \ascii{DKGCommit} run, it is every player's dominant strategy to comply with the protocol and so the run ends successfully, with every player obtaining a valid share per coin they deposited. These shares correspond to $x$, and $X=g^x$ is published as part of the run.

During a subsequent run of \ascii{DKGReveal}, players broadcast their shares to each other. Any player that collects $t+1$ shares can then reconstruct $x$ and publish it to the escrow. The escrow verifies that the published value is indeed the secret key that corresponds to $X$ using the Boolean function $\ascii{ver}(x,X)$, which verifies that $X=g^x$.

Eth-DKG~\cite{Eth-DKG_Github} implements such an escrow as a smart contract on Ethereum. 

\subsection{Our EVR Source Realization}
\label{sec:realization}
Our EVR source uses \ascii{DKGcommit} and \ascii{DKGreveal}. The smart contract that we implement, denoted $\G$, is an extension of the escrow functionality in Escrow-DKG. Like Escrow-DKG, $\G$ is a service in $\Z$ (see \Cref{sec:service}).
In \Cref{algo:G_alg} we give the pseudo-code for $\G$. It proceeds through the following phases, as depicted in \Cref{fig:G-stages}:
\input{figures/EVRS_algo_2}

\subsubsection*{The registration phase}
During registration, users submit registration transactions that invoke the $\G\s.\ascii{register}(\mathit{acc}, \$1)$ function of Escrow-DKG (we use the $\$$ sign to denote coins). 

\subsubsection*{The commit phase}
$\G$ enters the commit phase when the application calls $\G\s.\ascii{comTrigger}(\mathit{cndApp})$. This function sets the parameters for \ascii{DKGcommit}: $n$ is the number of accounts that have registered, and accordingly, the total number of shares to be generated by a successful run (and the total number of coins deposited in $\G$ during the registration phase); $t=\frac{2n}{3}$ is the threshold parameter. Recall that we assume that $N$ players register, such that player $i$ deposits $a_i$ coins (so $\sum_{i=1}^N a_i = n$), and thus, ends up with $a_i$ private shares (in a successful run).  $\G\s.\ascii{comTrigger}(\mathit{cndApp})$ sets three additional parameters (that are not related to \ascii{DKGcommit}; \cref{alg:G:params,alg:G:cnd}): $P=n-t=\nicefrac{n}{3}$ is the bound on the illicit profit that the EVR can sustain; $\ell=n$ is the informing reward, explained below; and $\mathit{cnd}=\mathit{cndApp}$ is the application's condition to trigger the \ascii{reveal} run. 

\ascii{DKGcommit} can now run among the $N$ players (and their $n$ accounts). Players invoke the \ascii{interactCommit} function to write information to $\G$'s state during the \ascii{DKGcommit} run. This information is vital for $\G$ to detect misbehavior (e.g., an account that does not supply the data it is expected to). The \ascii{DKGcommit} run has $t_{\text{com}}$ seconds to complete, by the end of which $\G$ verifies that no misbehavior was detected by any of the registered accounts (as described in the Escrow-DKG paper~\cite{rational_threshold_cryptosystems}); \cref{alg:G:no_misbehavior}. If $\G$ does not detect any misbehavior, $\G\s.\s\mathit{verCom}$ is updated to \true{}. Otherwise, $\G\s.\s\mathit{verCom}$ is updated to \false{} and $\G$ aborts.   

In a successful run, any $t+1$ shares (or more) can reconstruct $x$, while any $t$ shares (or less) reveal nothing about $x$. The run also computes $X=g^x$ and publishes it on $\G$. 

By \ref{def:evrs-usage:commit}, nothing can be gained from failing \ascii{DKGcommit}, and thus, when all players are rational, \ascii{DKGcommit} does not fail. Hence, our solution is non-trivial. In case it does fail, $\G$ returns all deposits that are tied to honest shares, as noted in \Cref{sec:escrow_DKG} above (this part is not shown in the pseudo-code).

\subsubsection*{The pending phase}
When the commit phase completes successfully, $\G$ enters the pending phase. This phase is the crux of the protocol, when only the informing function $\G\s.\ascii{inform}(x,\mathit{acc})$ can be called. Informing is key to ensuring secrecy, namely that $x$ remains secret so long as $\mathit{cnd}$ does not mature.

The incentive to break secrecy is evident -- a group of players who collude, secretly pass their shares among themselves, and manage to reconstruct $x$ before $\mathit{cnd}$ matures may make a substantial illicit profit by stealing up to $P$ coins (by \ref{def:evrs-usage:before_cnd}) from the application (for instance, by knowingly buying a winning lottery ticket). The informing mechanism nullifies this incentive by allowing any user who knows $x$ before $\mathit{cnd}$ matures to publish it on $\G$'s account for a high reward of $\ell$ coins (\cref{alg:G:informing_reward}). The informant who ``betrays'' the collusion is rewarded for her actions. The reward comes from the players' deposits, which are all confiscated in this case. If informing occurs, it implies that $t+1$ (or more) accounts have reconstructed $x$ while it was forbidden to do so; informing can thus be seen as a collective punishment of the players.

If multiple informants attempt to inform, we assume for simplicity that all of their informing transactions are submitted to $\Z$ simultaneously. $\Z$'s Log then orders the transactions according to the deterministic order described in \Cref{smart_contract_abstraction} and only the first transaction takes effect. If informing happens, $\G$ advances to the abort phase; \cref{alg:G:abort_informing}. Otherwise, $\G$ enters the reveal phase when $\mathit{cnd}$ matures; \cref{alg:G:enter_reveal}.

\subsubsection*{The reveal phase}
During the reveal phase, the players engage in a run of \ascii{DKGreveal} to reconstruct $x$. If the run completes successfully, then $x$ is published on $\G$. $\G$ then verifies that the published secret is correct (i.e., that it corresponds to $X$; \cref{alg:G:correct_secret}). If this is the case, $\G$ updates $\G\s.\s\mathit{verRev}$ to \true{}. The players then get their deposits back (\cref{alg:G:return_deposits}), and $\G$ advances to the final phase.

In case the players fail to publish $x$ in a timely manner, $\G$ enters the abort phase and all players' deposits are confiscated (\cref{alg:G:if_reveal_timeout,alg:G:reveal_abort}). This is another form of collective punishment that $\G$ enforces. Note that in order for this to happen, players holding at least $n-t$ shares must refrain from participating in the \ascii{DKGreveal} run.

We note that in the \ascii{DKGreveal} run, players do not publish their individual shares on $\G$. Rather, they exchange shares among themselves off-chain, reconstruct $x$ off-chain, and eventually publish $x$ on $\G$. The reason to do these steps off-chain is mostly a practical one. See \Cref{sec:implementation} for more details.  

\subsubsection*{The final and abort phases}
In the final phase, the application, or anyone for that matter, can read $x$ from $\G$. If $\G$ enters the final phase, then robustness is satisfied. Conversely, the abort phase implies that the EVR source has failed. In this phase, any call to any of $\G$'s functions fails.

\begin{figure}
\centering
\includegraphics[width=0.7\columnwidth]{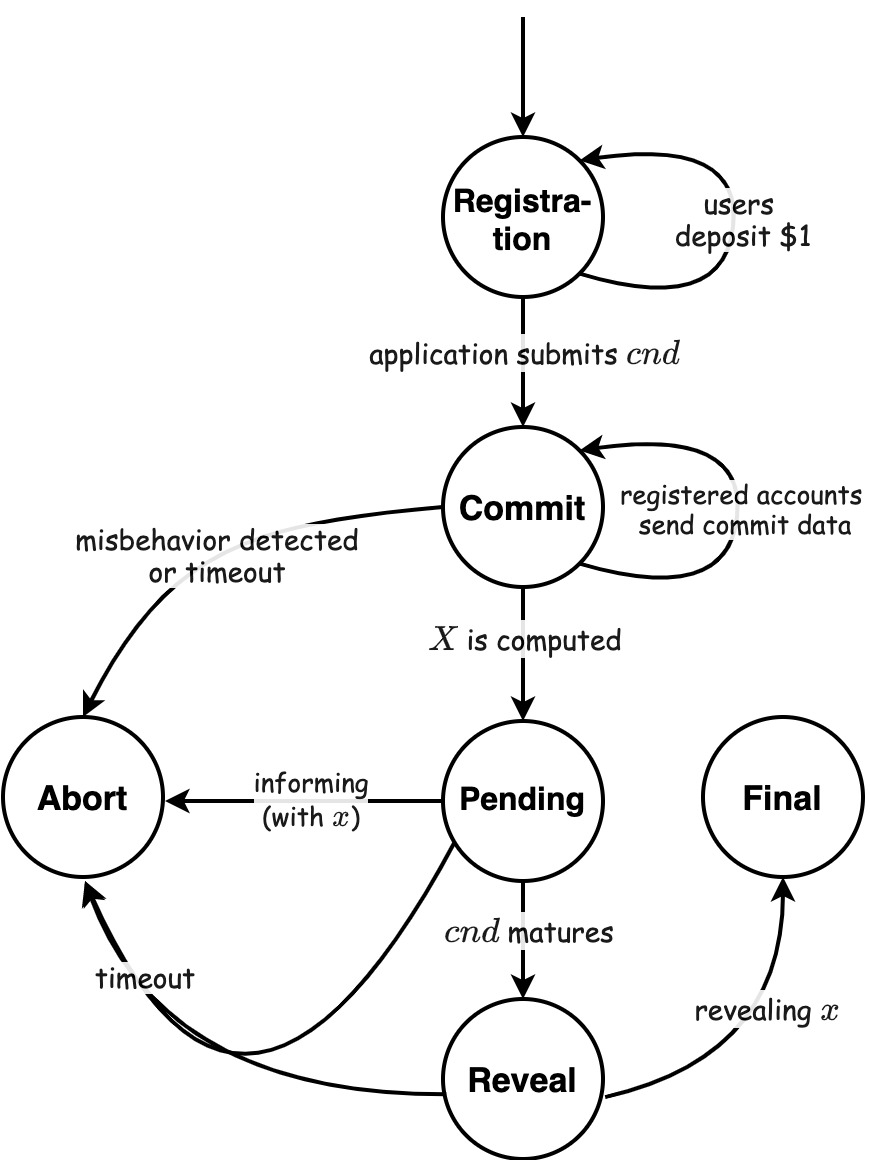}
\caption{A high-level view of $\G$'s phases.}
\label{fig:G-stages}
\end{figure}



\subsection{\texorpdfstring{$\G$}{G}'s Design Rationale}
\label{sec:design-rationale} 
We now give the intuition behind our parameters' choices: $t=\frac{2n}{3}$, $P=n-t$ and $\ell=n$.

First note that if $P > n-t$, then players holding $n-t$ shares might be better off failing $\ascii{DKGreveal}$ because, by \ref{def:evrs-usage:reveal}, they can potentially gain $P$ coins from the application while only losing $n-t$ coins to $\G$. Since we want $P$ to be as large as possible, we set $P=n-t$.

We would also like the informing reward, $\ell$, to be as large as possible, encouraging a player who obtains $t+1$ shares or more before $\mathit{cnd}$ matures to inform. We therefore set $\ell=n$, which is the total amount of coins in possession of $\G$.

By \ref{def:evrs-usage:after_cnd}, informing cannot be bypassed: in order to steal the application's coins, at least one player needs to locally reconstruct $x$. That player then needs to make a decision -- to inform or to steal. As long as there is no player who is better off stealing than informing, the informing mechanism is effective. It discourages players from attempting to break secrecy as they realize that any one of them would inform if they had the possibility to do so.

We note that side contracts cannot force an informant to forfeit or re-distribute her reward from $\G$. To clarify this subtle point, consider the following scenario: A player who wishes to steal the jackpot might attempt to convince fellow players to cooperate with her (i.e., send her their shares) by committing to give up the informing reward were she to inform. If that was possible, informing would lose its sting. This issue is solved by letting informants indicate an arbitrary account to transfer the reward to. Thus, an informant can always use a fresh EOA in $\Z$ that is not related to any side contract.

We are left with setting $t$. While an informant cannot distribute her informing reward, she can distribute her external coins. Additionally, in case of informing (which implies that robustness is violated), the application's coins might end up in players' hands by \ref{def:evrs-usage:reveal} (according to the application's fallback distribution). To make sure that informing results in losses for at least one colluding player, it must be the case that it is not profitable to distribute $e_i + P=e_i+n-t$ coins among players holding $t+1-a_i$ shares. The decentralization assumption (\Cref{eqn:high-participation}) takes care of this: $e_i+a_i \le 2t-n$ for all $i\in [N]$. We set $t$ as to maximize both the economic worth of the randomness $P=n-t$, and the maximal player balance allowed $e_i+a_i = 2t-n$. This leads us to set: $t=\frac{2n}{3}$.

Note that since we assume that a single player can register multiple times in $\G$, assuring secrecy would be hopeless if we did not assume that for all $i\in [N]$, $a_i\le t=\frac{2n}{3}$. The decentralization assumption is a bit stronger: $\forall i\in [N]$, $a_i+e_i\le \nicefrac{n}{3} = \nicefrac{t}{2}$.



\begin{remark} [Collective punishments]
Note that in case informing happens (which implies that secrecy breaks) or if the reveal phase timeouts (which implies that robustness breaks), all players are slashed including ones that behaved honestly. However, in case non-triviality breaks, no collective punishment is enforced.

Collective punishment is inevitable in case secrecy is violated as it is impossible to detect which players colluded in order to break secrecy. Regarding robustness, we chose to take the collective punishment approach due to practical considerations (see \Cref{sec:implementation} for details).
\end{remark}

%% file: figures/EVRS_algo_2.tex
\begin{algorithm} 
	\captionof{algorithm}{\label{algo:G_alg} The service $\G$}
	\begin{algorithmic}[1]
        
    	\footnotesize
    
        \Statex \textbf{Constants}: $t_{\text{com}}, t_{\text{rev}}$, $g$
    
        \Statex\textbf{Globals}: $X$, $x$ ,$P$ ,$\ell$ ,$t$ ,$\mathit{cnd}$ ,$\mathit{commitStart}$
        
        \Statex \textbf{Init}: 
        $\mathit{accounts} \leftarrow $ empty list;\hskip1em $n\leftarrow 0$;\hskip1em $\mathit{verCom}$ ,$\mathit{verRev}\leftarrow \perp$;\hskip1em 
        \Statex \hskip2.1em $\mathit{phase} \leftarrow \text{`registration'}$
        \Statex

        \Statex \underline{function $\ascii{ver}(x,X)$}
        \Indent
            \State \textbf{return} $X = g^x$
        \EndIndent
        
        
        \Statex \underline{function $\ascii{register}(\mathit{acc},\$1)$} 
        \Indent
            \State \textbf{if} $\mathit{phase} \neq \text{`registration'}$ \textbf{then} \textbf{exit}
            \State $\mathit{accounts}.\mathit{append}(\mathit{acc})$
            \State $n \leftarrow n+1$;\hskip1em $\G\s.\mathit{balance} = \G\s.\mathit{balance} + \$1$ 
        \EndIndent
        
        \Statex \underline{function $\ascii{comTrigger}(\mathit{cndApp})$} \Comment{invoked only by the application}
        \Indent
            \State \textbf{if} $\mathit{phase} \neq \text{`registration'}$ \textbf{then} \textbf{exit}
            \State $\mathit{phase} \leftarrow \text{`commit'}$
            \State $\mathit{commitStart} \leftarrow \mathit{now}$
            \State $t \leftarrow \frac{2n}{3}$;\hskip1em $\ell \leftarrow n$;\hskip1em $P \leftarrow n - t$ \label{alg:G:params}
            \State $\mathit{cnd} \leftarrow \mathit{cndApp}$ \label{alg:G:cnd}
        \EndIndent
        
        
        \Statex \underline{function $\ascii{interactCommit}(\mathit{data})$} 
        \Indent
            \State \textbf{if} $\mathit{phase} \neq \text{`commit'}$ \textbf{then} \textbf{exit}
            \State \textbf{if} $\mathit{now} - \mathit{commitStart} < t_{\text{com}}$
            \Indent
                \State save $\mathit{data}$ in $\mathit{localState}$
            \EndIndent
            \State \textbf{if} $\mathit{now} - \mathit{commitStart} \ge t_{\text{com}}$
            \Indent
                \State $\mathit{verCom} \leftarrow \ascii{noMisbehavior}(\mathit{localState})$  \label{alg:G:no_misbehavior}
                \State \textbf{if} $\mathit{verCom}$
                \Indent
                    \State compute $X$ from $\mathit{localState}$
                    \State $\mathit{phase} \leftarrow \text{`pending'}$
                \EndIndent
                \State \textbf{else}
                \Indent
                    \State $\mathit{phase} \leftarrow \text{`abort'}$
                \EndIndent
            \EndIndent
        \EndIndent
        
        
        \Statex  \underline{function $\ascii{inform}(\mathit{xInform},\mathit{acc})$} 
        \Indent
            \State \textbf{if} $\mathit{phase} = \text{`commit'}$ \textbf{then} $\ascii{interactCommit}()$
            \State \textbf{if} $\mathit{phase} \neq \text{`pending'}$ \textbf{then} \textbf{exit}
            
            \State $\mathit{verRev} \leftarrow \ascii{ver}(\mathit{xInform},X)$
            \State \textbf{if} (not $\mathit{cnd}$) \textbf{and} $\mathit{verRev}$
            \Indent
                \State $x \leftarrow \mathit{xInform}$
                \State pay $\mathit{acc}$ $\$\ell$;\hskip1em $\G\s.\mathit{balance} = \G\s.\mathit{balance} - \$\ell$  \label{alg:G:informing_reward}
                \State $\mathit{phase} \leftarrow \text{`abort'}$ \label{alg:G:abort_informing}
            \EndIndent
        \EndIndent
        
        
        \Statex  \underline{function $\ascii{reveal}(\mathit{xReveal})$} \label{alg:G:getrand}
        \Indent
            \State \textbf{if} $\mathit{phase} = \text{`commit'}$ \textbf{then} $\ascii{interactCommit}()$
            \State \textbf{if} $\mathit{phase} \neq \text{`pending'}$ \textbf{then} \textbf{exit}
            
            \State \textbf{if} $\mathit{cnd}$ \textbf{and} $\mathit{now}-\mathit{matureTime}(\mathit{cnd}) < t_{\text{rev}}$ \Comment{reveal phase (implicit)} \label{alg:G:enter_reveal}
            \Indent
                \State $\mathit{verRev} \leftarrow \ascii{ver}(\mathit{xReveal},X)$ \label{alg:G:correct_secret}
                \State \textbf{if} $\mathit{verRev}$ 
                \Indent
                    \State $x \leftarrow \mathit{xReveal}$
                    \State pay each $\mathit{acc} \in \mathit{accounts}$ $\$1$;\hskip1em $\G\s.\mathit{balance} = \G\s.\mathit{balance} - \$n$ \label{alg:G:return_deposits}
                    \State $\mathit{phase} \leftarrow \text{`final'}$
                \EndIndent
            \EndIndent
            \State \textbf{if} $\mathit{now}-\mathit{matureTime}(\mathit{cnd}) \geq t_{\text{rev}}$ \label{alg:G:if_reveal_timeout}
            \Indent
                \State $\mathit{phase} \leftarrow \text{`abort'}$ \label{alg:G:reveal_abort}
            \EndIndent
        \EndIndent
 	\end{algorithmic}
\end{algorithm}

%% file: analysis.tex
We now turn to analyze our EVR source. In \Cref{ssec:game-theory} we give necessary concepts in game theory. In \Cref{sec:game-def} we model our protocol as a game among $\G$'s players. Finally, \Cref{ssec:game-analysis} presents our analysis, showing that our solution satisfies the EVR properties under a strong game-theoretic solution concept.

\subsection{Preliminaries -- Game Theory}
\label{ssec:game-theory}





An $N$-player (normal-form) game $\game$ is a pair $\langle \prod_{i=1}^N S_i,\{u_i\}_{i=1}^N \rangle$, where $S_i$ is player $i$'s \emph{strategy set}; and $u_i:\prod_{i=1}^N S_i \to \mathbb{Z}$ is $i$'s \emph{payoff function}. We denote $S=\prod_{i=1}^N S_i$. We consider non-cooperative games where binding agreements are not possible.

Given a proper subset $J\subsetneq [N]$ and a strategy vector $s=\break(s_1,\dots,s_N) \in S$, denote by $s_J=(s_i)_{i \in J}$ the projection of $s$ on indexes in $J$,  and let $-J = [N] \setminus J$. In a game $\game$, when a subset of the players $-J$ fixes its strategy vector to some $s'_{-J}\in \prod_{i\in-J}S_i$, it induces a new game among the remaining players $J$, which we denote by $\game / s'_{-J}$. Thus, $\game / s'_{-J} \triangleq \langle \prod_{i\in J} S_i, \{\overline{u}_{i}\}_{i\in J} \rangle$, where $\overline{u}_{i}(s_J)=u_i(s_J,s'_{-J})$ for all $i\in J$ and $s_J\in \prod_{i\in J} S_i$.

A common solution concept in game theory is a Nash Equilibrium,
which stipulates that unilateral deviations are not profitable to individual players. This is a fairly weak solution concept, as a coalition might still gain from deviating. The concept of Strong Nash Equilibrium rules out deviations by any conceivable coalition.
However, this concept is often considered to be too strong, as explained in \cite{CPNE}: ``coalitions are allowed complete freedom in choosing their joint deviations: while the whole set of players is concerned with arriving at a strategy vector that is immune to deviations by any coalition, no deviating group of players faces a similar restriction.''

The concept of Coalition-Proof Nash Equilibrium, due to Bernheim et al.\ \cite{CPNE}, solves this inconsistency by considering only self-enforcing deviations, namely, deviations that are stable in the sense that no subset of the deviators has motivation to deviate further. Formally:
\begin{definition}[Coalition-Proof Nash equilibrium (CPNE)]
Let $N>0$ and $\game=\langle S,\{u_i\}_{i=1}^N \rangle$ be an $N$-player game.
\begin{enumerate}
    \item If $N=1$, strategy $s^*\in S$ is a CPNE in $\game$ if $s^*$ maximizes $u_1(\cdot)$.
    
    \item If $N>1$, assume that CPNE has been defined for games with fewer than $N$ players. Then:
    \begin{enumerate}
        \item A strategy $s^* \in S$ is \emph{self-enforcing} in $\game$ if for all $J \subsetneq [N]$, $s_J^*$ is a CPNE in $\game / s_{-J}^*$.
        
        \item A strategy $s^*\in S$ is a CPNE in $\game$ if it is self-enforcing in $\game$ and there does not exist another self-enforcing strategy $s\in S$ such that $u_i(s) > u_i(s^*)$ for all $i \in [N]$.
    \end{enumerate}
\end{enumerate}
\end{definition}

A CPNE is a powerful solution concept attainable in non-\break cooperative games. When considering only self-enforcing strategies, it ensures Pareto efficiency, namely, there is no other self-enforcing strategy vector that increases at least one player's payoff without decreasing anyone else's.

In general, proving that some strategy is a CPNE can be difficult. However, to the purpose of the upcoming analysis the following observation suffices.


\begin{observation}\label{observation:CPNE}
Let $\Gamma = \langle S,\{u_i\} \rangle$ and consider a strategy vector $s^*\in S$. Assume that for every deviating coalition $C\subset [N]$ and for every strategy vector $s^1_C\in \prod_{i\in C} S_i$, there exists a member $j\in C$ such that, either (i) $u_j(s^1_C,s^*_{-C}) \le u_j(s^*)$ (namely, $j$ in not better-off after the deviation), or (ii) there exists $s^2_j\in S_j$ such that $\overline{u}_j(s^2_j,s^1_{-j}) > \overline{u}_j(s^1_C)$ in $\game/s^*_{-C}$ (namely, $j$ is better-off unilaterally re-deviating from $s^1_C$, rendering the original deviation not self-enforcing). Then, $s^*$ is a CPNE in $\Gamma$.  
\end{observation}

\subsection{The Players' Game \texorpdfstring{$\game$}{}}
\label{sec:game-def}
Our game-theoretic analysis begins after a set of $N$ players have chosen to register in $\G$. (As noted in \Cref{sec:service}, incentivizing players to register can be done as in other services~\cite{CosmosWhitePaper,TezosWhitepaper,helix,AugurWhitepaper,TrueBit,LivepeerWhitepaper} and is beyond the scope of this work.) Specifically, player $i \in [N]$ has chosen to deposit $a_i$ coins and keep her remaining $e_i$ coins externally. We assume that the distribution of coins adheres to \Cref{eqn:high-participation} (decentralization).

We define $\game$, the $N$-player game played by $\G$'s players. We describe $\game$'s strategies in three consecutive stages, 
where in later stages players are aware of previous stages' outcome and make decisions in accordance (such a game description is sometimes referred to as extensive-form). 
\begin{enumerate}
    \item In stage 0 players engage in side contracts -- for each coin that they register with (and obtain a private share for), they decide whether to register with an EOA or via a side contract; additionally, they decide whether to use their external coins within side contracts. This stage is played during $\G$'s registration phase.
    \item In stage 1 each player $i$ selects a subset of $[N] \setminus \{i\}$ and sends to each member in this set some or all of her shares. Particularly, $i$ can choose to keep her shares private. This stage is played during $\G$'s pending phase.
    \item In stage 2 a player obtaining $t+1$ shares or more can inform or attempt to illicitly gain application's coins by stealing. Also, after $\mathit{cnd}$ matures, when \ascii{DKGreveal} runs, each player either complies with the protocol's rules or deviates. Players who comply with \ascii{DKGreveal} broadcast (off-chain) all of their shares, and players who deviate broadcast none of them. Once a compliant player collects $t+1$ shares, she reconstructs $x$ and publishes it (on $\G$). 
\end{enumerate}

Note that each stage is played instantaneously, namely, all players decide on their actions at the same time. Specifically, in stage 2, this means that a player $i$ who decides not to reveal her shares cannot decide to reveal them after seeing shares sent by other players (and potentially reconstructing $x$ privately with her own shares that she has yet to publish). This formulation has $i$ decide whether to reveal her shares \emph{before} she knows $x$. Indeed, our game models the players as pessimistic, namely, $i$ makes her decision under the belief that her legitimate $x$-dependent profit is zero. And yet, in practice, the value of $x$ might impact $i$'s legitimate profit from the application (e.g., if she is genuinely lucky to hold a winning ticket), thereby affecting her decision. 
Nevertheless, if $i$ is optimistic, she is ever more motivated to comply with \ascii{DKGreveal} in stage 2. Thus, ruling out such players (i.e., assuming pessimistic ones) does not weaken our analysis. 



The \emph{default} strategy is to follow the protocol: not to engage in side contracts in stage 0; not to send any shares in stage 1; and in stage 2, to inform on obtaining $t+1$ shares or more, and otherwise to broadcast all shares during \ascii{DKGreveal}, and publish $x$ on $\G$ whenever possible. 

$\game$'s payoffs, $\{ u_i \}_{i=1}^N$, are distributed at the end of the game, namely after stage 2. The payoff functions reflect the game's impact on the players' balances; thus,  player $i$'s total balance (in all her accounts) in $\Z$ at the end of the game is $e_i+u_i(s)$.
The payoffs depend on the game's high-level outcome in $s\in S$, as captured by the following three predicates:
\begin{enumerate}
    \item \texttt{INF}$(s)$ indicates whether informing occurs in $s$.
    \item \texttt{SEC}$(s)$ indicates whether secrecy is maintained in $s$, namely, if no player obtains $x$ before $\mathit{cnd}$ matures.
    \item \texttt{ROB}$(s)$ indicates whether robustness holds in $s$, i.e., \texttt{ROB}$(s)$ if $\G$ ends in the final phase (rather than abort).
\end{enumerate}
 
Note that $\texttt{INF}(s)$ implies $\neg \texttt{SEC}(s)$, as informing requires at least one player to obtain $x$ during stage 1. Also, $\texttt{INF}(s)$ implies $\neg \texttt{ROB}(s)$, because $\G$ enters the abort phase after informing. 

\newlist{myList}{enumerate*}{1}

There are three pools of coins that players may gain coins from or lose coins to. We use three auxiliary functions to define how the coins in these pools are distributed and postulate bounds on their values given the predicates above:
\begin{description}

    \item[Application's coins] 
    Recall that we model the players' beliefs\break about legitimate profits from the application as zero (as explained above). Thus, only illicit profits from the application are considered. We use $z:S \rightarrow \mathbb{Z}^N$ to capture such profits.
    
    The following bounds hold:
    \begin{enumerate}[label=A\arabic*.,ref=A\arabic*]
        \item\label{consumerCoins:sum_less_P} For all $i\in [N]$ and for all $s\in S$, $z_i(s) \ge 0$ and $\sum_{j=1}^N z_j(s) < P$, by \ref{def:evrs-usage:reveal} and \ref{def:evrs-usage:before_cnd}.
        
        \item\label{consumerCoins:SEC-ROB} If $\texttt{SEC}(s) \land \texttt{ROB}(s)$, then $\forall i\in [N]: z_i(s) = 0$, by \ref{def:evrs-usage:after_cnd}.
    \end{enumerate}
    
    
    
    \item[External coins] $y:S \rightarrow \mathbb{Z}^N$ determines how the players' external coins are distributed due to side contracts that burn or transfer the coins in certain circumstances. 
    
    We have the following bounds:
    \begin{enumerate}[label=E\arabic*.,ref=E\arabic*]
        \item\label{externalCoins:bounded_by_own_EC} For all $i\in [N]$ and for all $s\in S$, $y_i(s) \ge -e_i$ is restricted by $i$'s external coins.
        
        \item\label{externalCoins:sum-non-positive} Furthermore, for all $s\in S$, $\sum_{i=1}^N y_i(s) \le 0$.
        
        \item\label{externalCoins:no_side_contract} In every strategy vector $s\in S$ in which player $i\in [N]$ does not engage in side contracts in stage 0, $y_i(s) \ge 0$.
    \end{enumerate}
    
    
    
    \item[Deposits in $\G$] $w:S \rightarrow \mathbb{Z}^N$ determines how $\G$'s deposits are distributed among the players. Players can register through side contracts that, in certain circumstances, force them to redistribute their deposits if $\G$ pays them back.
    
    Hence, the following bounds hold:
    \begin{enumerate}[label=D\arabic*.,ref=D\arabic*]
        \item\label{depositCoins:sum} For all $i\in [N]$ and for all $s\in S$, $w_i(s) \ge 0$, and\break  $\sum_{j=1}^N w_j(s) \le n$.
        
        \item\label{depositCoins:informing} If $\texttt{INF}(s)$ then the informant, $f$, gains $\ell$ coins and all others forfeit their deposits. Namely, $w_f(s) = \ell$ and $\forall i\ne f : w_i(s)= 0$. (Recall from \Cref{sec:design-rationale} that an informant cannot bind herself, through side contracts or otherwise, to forfeit or re-distribute the informing reward.)
        
        \item\label{depositCoins:confiscated} If $\neg \texttt{INF}(s) \land \neg \texttt{ROB}(s)$, then $\forall i\in [N] : w_i(s)=0$ (i.e., all deposits are confiscated).
        
        \item\label{depositCoins:no_side_contract}  Also, in a strategy vector $s\in S$ in which player $i$ does not engage in side contracts in stage 0 and $\texttt{ROB}(s)$, $w_i(s) \ge a_i$.
    \end{enumerate}
    
\end{description}

\begin{remark}
We note that the functions $w(\cdot)$ and $y(\cdot)$ do not depend on $x$. This implies that $\game$ does not capture side bets between the players that depend on the value of $x$. Indeed, if such side bets were to be made then the randomness would be worth more than dictated by correct usage.
\end{remark}

Using these functions, we get $u_i(s)=z_i(s)+y_i(s)+w_i(s)$. We note that each player's payoff is at least $-e_i$ and the sum of players' payoffs is less than $n+P$, namely, 
\begin{enumerate}[label=P\arabic*.,ref=P\arabic*]
    \item\label{payoff:sum_less_than} $\forall s \in S$, $\forall i\in [N]:$ $u_i(s) \ge -e_i$ and  $\sum_{j=1}^N u_j(s) < n + P$.
\end{enumerate}



For brevity, we omit the strategy vector $s$ when it is clear from the context.

\subsection{Analysis}
\label{ssec:game-analysis}
Broadly speaking, our goal in this section is to show that in an economic game-theoretic setting, and under our assumptions regarding the players and correct usage, our solution for the EVR source yields strategy choices that satisfy the non-triviality, secrecy, and robustness properties. 
Non-triviality is immediate from \ref{def:evrs-usage:commit} and Escrow-DKG's properties, as explained in \Cref{sec:realization} (note that $\game$ is played only if non-triviality is met). To show secrecy and robustness, we analyze $\game$. We prove that the default strategy vector, which indeed satisfies both properties, is a CPNE in $\game$.

Let $s^*$ be the default strategy vector.
\begin{claim}
For all $i\in [N]$, $u_i(s^*)=a_i$.
\end{claim}
\begin{proof}
The decentralization assumption implies $\texttt{SEC}(s^*)$ -- no player obtains at least $t+1$ shares in stage 1, therefore informing does not occur and no illicit profit is made by knowing $x$ in advance. Additionally, in $s^*$ all players provide their shares in stage 2, which implies $\texttt{ROB}(s^*)$, and so no illicit profit is made from the application's fallback distribution. Finally, since no side contracts are used in $s^*$, no player loses any of her external coins or deposits. Hence, by \ref{consumerCoins:SEC-ROB}, \ref{externalCoins:sum-non-positive}, \ref{externalCoins:no_side_contract}, \ref{depositCoins:sum} and \ref{depositCoins:no_side_contract} the payoffs are: $u_i(s^*)=a_i$, for all $i\in [N]$.
\end{proof}


For the analysis's sake, we partition the strategy space $S$ into $S_{+}=\{s\in S: \forall i\in [N], u_i(s)\ge 0\}$ and $S_{-}=\{s\in S :\: \exists i\in [N], u_i(s) < 0\}$. Note that $S_+ \cup S_- = S$ and $S_+ \cap S_- = \emptyset$. 
We use the following notation in the proofs: given a set of players $C \subseteq [N]$, $\tickets{C} \triangleq \sum_{i\in C} a_i$.

The following lemma states that $s^*$ is immune to deviations in $S_+$ that preserve secrecy, and is proven in \Cref{app:omitted-proofs}. 
\begin{lemma} \label{Lemma:SNE_robustness_secrecy}
For all $C\subseteq [N]$ and for all $s_C \in \prod_{i\in C} S_i$, such that (i) \texttt{SEC}$(s_{C},s^*_{-C})$ and (ii) $(s_{C},s^*_{-C})\in S_+$,
there exists a player $i\in C$ such that $u_i(s^*) \ge u_i(s_{C},s^*_{-C})$.
\end{lemma}

We proceed to analyze the game when secrecy is violated, and illicit profit and informing are feasible. The following lemma, proven in \Cref{app:omitted-proofs}, shows that whenever informing happens, at least one player who deviated from $s^*$ and facilitated informing is not better off.
\begin{lemma} \label{Lemma:informing_deviation}
For all $C\subseteq [N]$ and for all $s_C \in \prod_{i\in C} S_i$, such that (i) \texttt{INF}$(s_{C},s^*_{-C})$ and (ii) $(s_{C},s^*_{-C}) \in S_+$, 
there exists a player $i\in C$ such that $u_i(s^*) \ge u_i(s_{C},s^*_{-C})$.
\end{lemma}

The following lemma, proven in \Cref{app:omitted-proofs} using \Cref{Lemma:informing_deviation}, shows that any deviation from $s^*$ that breaks secrecy is either not self-enforcing or has (at least) one deviating player who is not better off. 
\begin{lemma} \label{Lemma:collusion_CPNE}
For all $C\subseteq [N]$ and all $s_C \in \prod_{i\in C} S_i$, such that (i) $\neg$\texttt{SEC}$(s_{C},s^*_{-C})$ and (ii) $(s_{C},s^*_{-C}) \in S_+$, either $(s_{C},s^*_{-C})$ is not self-enforcing, or
there exists a player $i\in C$ such that $u_i(s^*) \ge u_i(s_{C},s^*_{-C})$.
\end{lemma}

We note that in \Cref{Lemma:collusion_CPNE}, when $(s_C,s^*_{-C})$ is not self-enforcing, it is due to the possibility of one of $C$'s members to unilaterally re-deviate profitably. Finally, we prove our main result.
\begin{theorem} \label{Theorem:CPNE}
$s^*$ is a CPNE in $\game$.
\end{theorem}

\begin{proof}
Consider a coalition $C$ that deviates from $s^*$ resulting in a strategy vector $s=(s_C,s^*_{-C})$. 

If $s\in S_-$, then let $j\in [N]$ be a player for whom $u_j(s) < 0$. Surely, $j\in C$, as if $j$ does not engage in any side contracts (as in $s^*$), she cannot lose any of her external coins and is thus guaranteed a non-negative payoff. Thus, $j$ is a deviating player who loses from the deviation.

It remains to consider $s\in S_+$. Consider two cases:
\begin{enumerate}
    \item $\texttt{SEC}(s)$. From Lemma~\ref{Lemma:SNE_robustness_secrecy}, there exists a player $i\in C$ such that $u_i(s^*) \ge u_i(s)$, that is, $i$ is a member of $C$ that is not better off due to the deviation.
    
    \item $\neg \texttt{SEC}(s)$. From \Cref{Lemma:collusion_CPNE} either $s$ is not self-enforcing or there exists a player $j\in C$ who is not better off due to the deviation.
\end{enumerate}

    
    
We have shown that any possible deviation of a coalition $C\subseteq [N]$ from $s^*$ is either not self-enforcing (due to a unilateral deviation of one of $C$'s members) or includes a member that is not better off. Hence, by Observation~\ref{observation:CPNE}, $s^*$ is a CPNE in $\game$ as required.
\end{proof}

\begin{remark}
$s^*$ is not a Strong Nash Equilibrium in $\game$ because a coalition with $t+1$ shares that deviates, breaks secrecy, and illicitly gains $P$ coins may benefit all of its members. Nevertheless, as shown, informing renders such deviations not self-enforcing.
\end{remark}


%% file: multi-shot_ideal_source.tex

Up to this point, for simplicity, we focused on a randomness source that produces a single random value. In principle, such a source may be used also by applications that require a sequence of random values, such as card and die games. To this end, the application needs to invoke $\G.\ascii{comTrigger}$ for each random value it consumes. However, this approach is undesirable for two reasons. First, it is likely that the \ascii{commit} protocol is expensive and so we would like to avoid running it multiple times. In fact, in our EVR source realization, it is the only phase when disputes might need to be handled (as part of the DKG).

Second, recall that the EVR incentive structure explicitly leverages the fact that no illicit profit can be made by failing \ascii{commit} (\ref{def:evrs-usage:commit}). This is easy to establish when \ascii{commit} is executed once in a dApp, before the game begins and before the gamblers send their coins to the application's contract. But a dApp that guarantees a refund to gamblers cannot satisfy \ref{def:evrs-usage:commit} in its second \ascii{commit} run: by this point, some gamblers might discover, for instance, that they are likely to lose their gambling coins, creating an incentive for them to fail the \ascii{commit} run (which would force the application to refund the gamblers). In our solution, failing a \ascii{DKGcommit} run costs at most $1$ coin, so an about-to-lose gambler can register once in the second \ascii{DKGcommit} run and fail it. They would lose $1$ coin for failing \ascii{DKGcommit}, but would possibly save a much larger bet in the application.


To overcome this issue, we define and construct a \emph{multi-shot randomness source}, which robustly reveals a sequence of random values following a single \ascii{commit}. Our construction is based on \emph{verifiable random functions (VRFs)}. Background on VRFs is given in \Cref{VRF_background}, and the  multi-shot EVR source's definition is presented in \Cref{ssec:multi-shot_EVR}. The adaption of \Cref{algo:G_alg} to the multi-shot case is straightforward and its implementation is discussed in \Cref{sec:implementation} below. 

\subsection{Preliminaries -- Verifiable Random Functions}\label{VRF_background}
A VRF~\cite{VRF} is a local (non-distributed) function providing a pseudorandom value along with a cryptographic proof that the value was indeed randomly generated. 
We give a simplified definition of a VRF: 

\begin{definition} [VRF]
A VRF is a triple of efficient algorithms:
\begin{itemize}
    \item \ascii{gen}$()$, returning a pair of keys $(PK,SK)$;
    \item \ascii{eval}$(SK,m)$, returning a bit string $\sigma$;  and
    \item a Boolean function $\ascii{ver-VRF}(PK,m,\sigma)$. 
    
\end{itemize}
Assume \ascii{gen}$()$ returns $(PK,SK)$. Then the following hold:
\begin{description}
    \item[Complete Provability] For all bit strings $m$, if \ascii{eval}$(SK,m)=\sigma$, then $\ascii{ver-VRF}(PK,m,\sigma)=\true{}$.
    
    \item[Unique Provability] For all bit strings $m,\sigma_1,\sigma_2$, such that $\sigma_1 \ne \sigma_2$, either \ascii{ver-VRF}$(PK,m,\sigma_1)=\false{}$ or \ascii{ver-VRF}$(PK,m,\sigma_2)=\false{}$.\footnote{In the original definition, \ascii{ver-VRF} is a probabilistic algorithm and properties (1) and (2) hold with overwhelming probability in the security parameter.}
    \item[Residual Pseudorandomness] There does not exist an efficient algorithm that receives $PK$ as input and after selecting some $m$ (as it pleases) can distinguish (with non-negligible advantage) between $\sigma=\ascii{eval}(SK,m) \in \{0,1\}^h$ and a bit string sampled uniformly at random from $\{0,1\}^h$.
\end{description}
\end{definition}

VRFs can be implemented by certain signature schemes \cite{VRF}. We use a specific BLS signature scheme as a VRF in our prototype implementation (see \Cref{sec:implementation}). 

\subsection{Multi-Shot EVR Source} \label{ssec:multi-shot_EVR}
Let (\ascii{gen}, \ascii{eval}, \ascii{ver-VRF}) be a VRF. A \emph{multi-shot randomness source} consists of a pair of protocols, \ascii{commit} and \ascii{reveal}, and a smart contract $\E$ in $\Z$, as follows: 
\begin{itemize}
    \item A (successful) run of \ascii{commit} computes $\ascii{gen}()$, producing $(PK,SK)$, and publishes $PK$ on $\E$ ($SK$ remains obscure as implied from the secrecy property below).
    
    \item Following a \ascii{commit} run, the $i$th (successful) run of \ascii{reveal} computes $\sigma_i = \ascii{eval}(SK,i)$ and publishes it on $\E$.
    
    \item $\E$ exposes the following API:
    \begin{itemize}
        \item A function $\ascii{comTrigger}(\mathit{cnd}[k])$ called by the application that triggers a \ascii{commit} run, where $\mathit{cnd}[k]$ is a list of $k$ verifiable conditions in $\Z$, $\mathit{cnd}[k]=(\mathit{cnd}_1,\dots,\mathit{cnd}_k)$, determined by the application, that all eventually mature. Once $\mathit{cnd}_i$ matures, the $i$th \ascii{reveal} run begins (given that $i-1$ \ascii{reveal} runs have already completed).
    
        \item $k+1$ variables, $\mathit{verCom}$ and ({$\mathit{verRev}_1$},$\dots$,{$\mathit{verRev}_k$}), that indicate whether the \ascii{commit} and subsequent $k$ \ascii{reveal} runs (resp.) have succeeded. All variables are initiated to $\perp$ and later turn \true{} or \false{}. $\mathit{verRev}_i$ is updated using $\ascii{ver-VRF}(PK,i,\sigma_i)$; $\mathit{verCom}$ is protocol-specific, and (among other things) verifies that the pair $(PK,SK)$ is a valid key pair that \ascii{gen} could output.
    
        \item $k+1$ timeout constants, $t_{\text{com}}$ and $(t_{\text{rev}}^1,\dots,t_{\text{rev}}^k)$ that dictate the maximum time that the \ascii{commit} and subsequent $k$ \ascii{reveal} runs (resp.) can take. In case one of the runs does not complete in a timely manner, $\E$ sets the appropriate variable to \false{}.
    \end{itemize}
\end{itemize}

As in the single-shot case, an \emph{ideal multi-shot randomness source} is a multi-shot randomness source satisfies the following core properties:
\begin{description}
    \item[Non-triviality] If $\E.\ascii{comTrigger}(\mathit{cnd}[k]$) is called, then $\E.\s\mathit{verCom}$ turns \true{}.
    
    \item[Hiding secrecy] If $\E.\s\mathit{verCom}$ turns \true{}, then for all $1 \le i \le k$ as long as $\mathit{cnd}_i$ does not mature, no user obtains \emph{any} information about $\sigma_i$.

    \item[Robustness] If $\E.\s\mathit{verCom}$ turns \true{}, then after $\mathit{cnd}_i$ matures\break $\E.\s\mathit{verRev}_i$ also turns \true{}.
\end{description}

A multi-shot EVR source is a multi-shot randomness source, possibly
satisfying secrecy rather than hiding secrecy. It satisfies the core properties provided that the application satisfies the economic restrictions of \Cref{def:evrs-usage}.

In the ensuing section we spell out the adjustments required in order to turn our single-shot EVR source (\Cref{sec:EVR_realization}) into a multi-shot EVR source.


%% file: prototype.tex
We now present our implementation of a multi-shot variant of $\G$ over Ethereum, denoted $\hat{\G}$. In \Cref{sec:ETHvsZ} we point to a few aspects where Ethereum diverges from our idealized blockchain abstraction $\Z$. In \Cref{sec:impl_details} we provide implementation details, and in \Cref{sec:real-world}, we address issues that arise due to the aspects discussed in \Cref{sec:ETHvsZ}.

\subsection{``Mind The Gap'': Ethereum vs \texorpdfstring{$\Z$}{B}}
\label{sec:ETHvsZ}
As described in \Cref{sec:model}, $\Z$'s Log can instantaneously append and process an unbounded number of transactions that consume unbounded computational resources. If transactions are issued concurrently, they are appended to the Log according to some deterministic order. Conversely, in Ethereum, transactions are appended to the blockchain in blocks, which are mined at an average pace of $1$ block per $12-15$ seconds~\cite{ethereum_block_time}, and are bounded in the computational resources they may consume (and thus also in the number of transactions they can include). Specifically, blocks are allowed to consume at most $\approx 10^7$ \emph{gas}~\cite{etherscan_gas_limit} (as of June 2020), a unit of measurement quantifying the computational resources in Ethereum. These restrictions imply that transaction issuers compete for limited resources. The competition takes place in the form of a gas auction: transaction issuers bid how much fee they are willing to pay per unit of gas that their transaction consumes.
Among the pending transactions, miners have complete freedom to pick transactions to their blocks. Since a transaction's fee is paid to the miner that includes it in a block, miners tend to pick the highest fee-paying transactions.   

Ethereum's open fee market introduces the possibility to front-run: front-running occurs when a later issued transaction out-bids an earlier one and ends up being processed first. 

Whereas $\Z$'s Log is immutable, in Ethereum, miners choose voluntarily on top of which blockchain tip to mine, creating a possibility for chain \emph{reorgs}, where blocks that appear to have been appended are supplanted by other blocks (if hashrate is sufficiently decentralized and block propagation is fast enough, miners are incentivized to mine on top of the longest branch, which is considered to be the canonical one, and this risk is reduced).






\subsection{Implementation Details}
\label{sec:impl_details}
To realize Escrow-DKG in Ethereum, we use the open-source Eth-DKG library of Asayag et al.\ \cite{Eth-DKG_Github} (we chose to use this library rather than a similar one due to Schindler et al.\ \cite{EthDKG_Github_austrians} because the former incorporates deposits and slashing as needed for our game-theoretic setting). Loosely speaking, Eth-DKG implements a VRF via a threshold BLS signature scheme. During the commit phase, it generates a pair of keys -- $x$ the secret key for signing and $X=g^x$ the public key for signature verification. While $X$ is published on-chain, $x$, is shared by the participants via Shamir secret sharing. The library can be used to verify that key generation was successful and update $\mathit{verCom}$ as described in \Cref{sec:realization}. During the pending and reveal phases, signatures (rather than $x$ itself as in the single-shot version) are published on-chain. Every signature is evaluated via an on-chain signature verification function that uses $X$ (and updates $\mathit{verRev}_i$). The library utilizes a number of pre-compiled contracts (that were originally introduced to Ethereum to allow for efficient on-chain verification of zk-SNARKS \cite{ethereum,zkSNARKs})
for signature verification. A single signature verification consumes 113,000 gas \cite{pairing_gas_cost}.

The reveal phase is scalable: signature reconstruction happens off-chain and only verification is done on-chain, at a price of verifying a single (threshold) BLS signature, regardless of the number of shares that were required to compute the signature or the number of participants in the DKG. The off-chain reconstruction requires communication among the players. On a public WAN, it takes less than a second (in normal conditions, with hundreds of participants). The delay in the reveal phase thus mostly depends on Ethereum's block rate. To be on the safe side, we set the timeout in the reveal phase to be $10$ minutes (this parameter is configurable). Note that the timeout only affects latency in case of failure, whereas successful reveals can be fast.

The on-chain component of Eth-DKG's key generation (i.e., the commit phase) is less scalable. 
It essentially verifies $n$ independent Feldman verifiable secret sharing~\cite{FeldmanVSS} runs. Implementing this on-chain naïvely would consume a lot of resources and would be infeasible even for small $n$s. Eth-DKG is able to steer most of this computational burden off-chain by taking an optimistic approach: It proceeds assuming that all players follow the protocol and allows accounts to file disputes in case they detect a problem. In the latter case, dispute arbitration happens by an interactive protocol between the disputing players. The arbitration protocol has both on-chain and off-chain components. This way, Eth-DKG can accommodate a few hundreds of participants without exceeding the Ethereum block gas limit~\cite{rational_threshold_cryptosystems,EthDKG_austrians}. We set $n_{\text{max}}=256$ in $\hat{\G}$.

\subsection{Addressing Real-World Issues}
\label{sec:real-world}
\subsubsection{Registration}
The registration phase needs to make sure that no more than $n_{\text{max}}$ accounts register. But restricting $n$ risks the decentralization of the EVR source as a single entity that registers multiple times might prevent others from registering.
Our heuristic to address this issue is to leave the registration phase open for a predefined time (we set it to 24 hours, but this is configurable), and let users bid how much deposit they are willing to invest. By the end of the registration period, the top $n_{\text{max}}$ bidders are chosen to participate. 
The risk with this method is that at the very end of the registration phase an attacker can submit multiple bids, slightly outbidding the current ones, without time for others to respond. To mitigate this problem, we limit the number of accepted bids per block to 1 and only accept bids that outbid the current lowest bid by at least $3\%$ (again, these parameters are configurable).

While in $\G$ the registration ends when the application calls \ascii{comTrigger}, which it can do whenever it pleases, in $\hat{\G}$ we set a few constraints on the transition from the registration to the commit phases. As mentioned, we set a minimum time for the registration phase -- 24 hours. Also, we set $n_\text{min}=100$ such that as long as $n_\text{min}$ accounts do not register, registration cannot end. Finally, if \ascii{comTrigger} is not called within a predetermined time frame, the registered accounts can withdraw their deposits. Another slight variation between $\G$ and $\hat{\G}$ is that in our implementation, we allow multiple registrations from the same account in a single transaction in order to reduce gas costs.


\subsubsection{Informing}
The informing mechanism we propose in \Cref{sec:EVR_realization} is susceptible to front-running -- an attacker listening to the network may detect an informing transaction, learn the secret, submit a competing informing transaction with a higher gas bid, and front-run the original transaction to collect the informing reward. This risk nullifies the informant's original incentive to inform and breaks our EVR source's incentive layer. 

To address this problem we employ a two phase commit-reveal informing mechanism. In the first phase, the informant submits a commit-informing transaction that includes a hash of the secret $x$ and the account that is meant to receive the informing reward -- $\ascii{hash}(x||\mathit{acc})$. In the second phase, the informant reveals both $x$ and $\mathit{acc}$ in a reveal-informing transaction. 

From the commit-informing transaction, an attacker that owns an EOA $\mathit{acc}_A$ cannot infer a corresponding hash, $\ascii{hash}(x||\mathit{acc}_A)$, and cannot front-run the original informant. The only way for an attacker to compute $\ascii{hash}(x||\mathit{acc}_A)$ is by obtaining $x$. If the attacker learns $x$ only from the reveal transaction, then she needs to cause a reorg back to the block before the original informant's commit-informing transaction was included. By requiring a long enough delay between the commit-informing and the reveal-informing transactions, such a reorg is not a viable risk (we configure this delay to be 30 blocks).

To disincentivize false commit-informing transactions (where the presumed informant does not know $x$ and submits a garbage hash), $\hat{\G}$ processes such transactions only if they come with a deposit. The deposit is paid back when a corresponding reveal-informing transaction is processed. 

\subsubsection{Reconstructing $x$: on-chain vs off-chain}
In our EVR source, in the reveal phase, players engage in a run of \ascii{DKGreveal} that reconstructs $x$ (or a signature in the multi-shot variant) off-chain and then publishes it on-chain for $\G$ to verify. We could have taken a different approach, where players publish their individual shares on $\G$, which then calculates $x$ on-chain (once $t+1$ shares are published). This approach could eliminate the collective punishment enforced when robustness breaks, as $\G$ can now confiscate only the deposits tied to shares that were not revealed on time. However, this approach does not scale in the number of shares, and we therefore opt for off-chain reconstruction both in our theoretical solution and in our implementation.

\subsubsection{Transaction fees} 
Transaction fees in Ethereum are not predictable and vary dramatically between periods of high and low demand \cite{gas_price_etherscan}. 

When transaction fees are significant, the off-chain reconstruction introduces asymmetry among the players: a single player submits a reveal transaction and has to pay the fee, whereas other players do not pay anything. The risk is that in order to avoid this extra payment, players would leave it to other players to submit the reveal transaction, ending up in a situation where no one actually submits it.  

To mitigate this risk, we add to $\hat{\G}$ a built-in fee refund mechanism (only for reveal transactions). The refund is paid to the player who submits the transaction, with $\hat{\G}$'s coins. In this way, the fee is divided equally among all players: instead of returning the full 1 coin deposit (per registered account), $\hat{\G}$ returns $1-\frac{\emph{fee}}{n}$.

Nevertheless, a naïve refund mechanism that simply reimburses the full amount is susceptible to an attack where the player submitting the transaction bids for an unnecessarily high gas price. This attack benefits the miner who mines the transaction, and this miner might collude with the bidder.
To mitigate this risk, $\hat{\G}$'s reveal phase stays open for some short time (5 blocks) after the first reveal transaction is processed. During this time, players may submit additional reveal transactions paying lower fees. Only the player who pays the lowest fee is refunded. Thus, if players detect a reveal transaction that substantially exceeds the current fee level, they submit a cheaper reveal transaction, to make sure that they are not paying too much to miners.


Orthogonal to this technique, to reduce fees that players need to pay, we use a GasToken-based approach. GasToken~\cite{GasToken} allows users to (conceptually) purchase gas, store it, and later on consume it to cover transaction fees when gas prices are high. Our contract preemptively purchases gas tokens if it is cheap and automatically (without user intervention) uses them when gas prices are high (gas prices are considered high when they pass some threshold, measured according to historical data) for reveal transactions. 


%% file: related_work_2.tex
While distributed randomness generation has been widely studied for decades, our work is unique in considering the randomness generation problem in an economic context and providing a solution that motivates cooperation under a strong game-theoretic concept. We now discuss related approaches to randomness generation.

\paragraph{Traditional distributed coin flipping}
Traditionally, coin flipping protocols~\cite{CoinTossingBlum,Cleve86,MoranNaorSegev09,BeimelOmriOrlov10} are designed for an adversarial model where some threshold of the participating processes are Byzantine and the remaining ones follow the protocol.  
From such protocols, we adopt the approach of committing to randomly sampled secrets before revealing them; in particular, our multi-shot EVR source uses VRFs~\cite{VRF}, as previously done in~\cite{Chainlink_VRF,SBRDR,Not-COINcidence,AlgorandV9,Dfinity,polkadot,keep_random_beacon}. Note that protocols designed for adversarial models are not applicable as-is to economic settings due to economic pressures that might lead more than the predefined threshold of the participants to diverge from the protocol.
Our solution mitigates this problem using secret-sharing and incentives, as well as bounding the economic worth of the secret. We prove that within our economic context, such attacks are not profitable and thus are not exercised by our game-theoretic players.

\paragraph{Blockchain-based randomness}
Similarly to our EVR source, a number of recent works have used public blockchains for distributed randomness generation, where the blockchain provides the source of truth regarding the generated random values. 

Some of these works consider a fully-adversarial setting, where \emph{any number} of players may be malicious~\cite{bitcoin_fair_lottery, MPC_bitcoin, zero_collateral_lotteries}. Although these solutions are more robust than ours (as they consider a broader range of user behaviors), this comes at a cost: In Bitcoin-based lotteries~\cite{bitcoin_fair_lottery, MPC_bitcoin}, deposits are very high -- $O(n^2)$ where the jackpot and the number of participants are both $O(n)$. In contrast, our deposits are $1$ coin for a jackpot of $O(n)$ with $n$ participants. And while in Zero-Collateral Lotteries~\cite{zero_collateral_lotteries} there are no deposits, they have limited scalability, as they require players to actively interact on-chain in $O(\log n)$ rounds. The economic setting in which our solution is realized enable us to reduce both the on-chain load and the deposits, while maintaining the trustworthiness of the randomness produced.
Moreover, in the aforementioned works, the randomness generation is tied to a particular lottery application, whereas our solution provides randomness as-a-service, to arbitrary applications. 

Similarly to our EVR source, RANDAO~\cite{RandaoPaperV085, RandaoGithub} is a smart contract on Ethereum that 
is used as an escrow to incentivize correct execution of a commit-reveal scheme. Yet RANDAO has not modeled user behavior or considered the economic implications of the randomness it produces. If it were to be used for a lottery in a model similar to ours, it would require deposits of $O(n^2)$ similarly to~\cite{bitcoin_fair_lottery, MPC_bitcoin}, because its incentive structure is similar to theirs. 

The blockchain itself (e.g., Nakamoto consensus~\cite{Bitcoin} or Algorand~\cite{AlgorandV9, AlgorandImpl}) typically generates and uses randomness for the specific purpose of selecting a leader to append new blocks to the chain. Like our protocol, blockchains require decentralization in order to work correctly~\cite{SelfishMining, SelfishMining-AvivZohar, Stubborn_Mining}. Yet unlike blockchains, our protocol preserves decentralization if it exists in the initial state -- the informing mechanism acts as an effective counter-measure to centralization pressure. 

Some recent works have exploited Bitcoin's proofs-of-work to extract publicly-verifiable random bits~\cite{bitcoin_beacon_LIFs,bitcoin_public_randomness} that are safe to use by general lotteries. 
Like our work, (and in contrast to most other works in the area), 
they consider selfish agents that attempt to distort the randomness generation process for a profit. 
Specifically, they show under which conditions it is profitable for a lottery player to attack. However, a user partaking in the lottery has no way of knowing whether other players can attack profitably, and so cannot tell whether the random generation process is indeed fair. In contrast, our solution ensures the fairness of the process whenever the application follows the correct usage guidelines. The EVR source explicitly and publicly specifies the bound on the randomness' worth, and any user can verify (using the public blockchain) that the application respects these bounds.

\paragraph{Related cryptographic primitives}
A game-theoretic model was also considered in the context of \emph{Rational Secret Sharing (RSS)}~\cite{RationalSecretSharingHalpern,RationalSecretSharingGordon,RationalSecretSharingLindell,RationalSecretSharingNonSimultanous1,RationalSecretSharingNonSimultanous2}, which focuses on reconstructing a secret (not necessarily a random value) among a network of selfish agents who prefer to learn the secret alone rather than together. In contrast, in our setting, the main profit is made when $x$ is publicly published, and so learning it alone (after $\mathit{cnd}$ matures) has little benefit. 


Delayed computation techniques~\cite{random_zoo_sloth_unicorn} utilize inherently sequential computations in order to generate random bits. A promising research direction in this vein is \emph{verifiable delay functions (VDFs)}~\cite{Boneh_VDF,Pietrzak_VDF,Wesolowski_VDF}, which are exponentially faster to verify than to evaluate, and hence, their input can be used as a commitment to a random value that will take time to be discovered. However, as of today, VDFs are not readily usable without some trusted setup. Furthermore, delays vary dramatically across hardware technologies, implying that the random value is revealed much earlier to some users than to others, making it difficult to incentivize slower users to partake in the evaluation process.
Whereas VDFs are not practical today, simple delay functions (without fast verifiability) have been implemented using smart contracts~\cite{proof_of_delay_ethereum}. But their verification occurs via an elaborate dispute process that consumes significant on-chain resources, and can take place for every new random value. In contrast, in our multi-shot EVR source, such a dispute process might take place at most once -- during the commit phase. Subsequently, verification of new random values occurs exclusively on-chain. Moreover, in the delayed computation approach there is an inherent delay before every new random value is publicized, whereas our EVR source can produce random bits on-demand within a latency proportional to the blockchain's block generation rate.


%% file: conclusion.tex
Blockchain-based dApps are proliferating nowadays, offering a wide range of decentralized trusted services. A lucrative application domain in this context 
is online gambling. Nevertheless, dApps implemented as smart contracts are inherently deterministic and thus cannot natively support such dApps. Rather, they need a trusted external source of randomness. In this work, we have addressed this need for a randomness source that can sustain its trustworthiness even when the provided randomness has significant economic consequences. To this end, we introduced the notion of EVR and showed how to build an EVR source.

Our EVR source produces random bits via an open distributed protocol where players are rational and may collude using side-contracts on the blockchain. 
Our protocol incorporates a  novel \emph{informing} mechanism, 
which acts as an effective deterrent against collusion, guaranteeing that the secret bits indeed remain secret. 

Our game-theoretic analysis has shown that as long as none of the players is ``too rich'', following the default strategy gives rise to a Coalition-Proof Nash Equilibrium -- a powerful solution concept in game theory -- where secrecy and robustness hold.

We implemented a proof-of-concept of our EVR source as a smart contract over the Ethereum public blockchain, optimizing off-chain communication to achieve scalability to hundreds of players. 
We hope that future work will suggest -- and build -- additional realizations of the EVR source formalized herein. In particular, it would be interesting to push scalability even further, perhaps by improving the off-chain execution path.

%% file: proofs/Lemma_1.tex
\begin{lemmaclone} {Lemma:SNE_robustness_secrecy}
For all $C\subseteq [N]$ and for all $s_C \in \prod_{i\in C} S_i$, such that (i) SEC$(s_{C},s^*_{-C})$ and (ii) $(s_{C},s^*_{-C})\in S_+$,
there exists a player $i\in C$ such that $u_i(s^*) \ge u_i(s_{C},s^*_{-C})$.
\end{lemmaclone}

\begin{proof}
Consider a coalition $C$ that deviates from $s^*$ resulting in a strategy vector $s \triangleq (s_{C},s^*_{-C})$ as assumed. 

Consider the case of \texttt{ROB}$(s)$, then by \ref{consumerCoins:SEC-ROB} $\forall i\in[N]: z_i(s)=0$, so $u_i(s) = y_i(s) + w_i(s)$. 
Note that $\sum_{i=1}^N (y_i+w_i-a_i) \le 0$, because $\sum_{i=1}^N y_i \le 0$ from \ref{externalCoins:sum-non-positive}, $\sum_{i=1}^N w_i \le n$ from \ref{depositCoins:sum} and $\sum_{i=1}^N a_i = n$ from definition.
If for all $i\in[N] : y_i+w_i-a_i=0$, then $u_i(s)=a_i$ thereby, none of the payoffs change relative to $s^*$.
Otherwise, there is some $i\in[N]$ such that $y_i+w_i-a_i\ne 0$. Hence, there is at least one player $j$ where $y_j+w_j-a_j < 0$, therefore, $u_j(s)<u_j(s^*)=a_j$. If $j\notin C$, then $j$ does not engage in side contracts. This implies that $u_j(s)\ge a_j$ (by \ref{externalCoins:no_side_contract} and \ref{depositCoins:no_side_contract}), which is a contradiction. Thus, $j\in C$ as required for this case.



We move to consider the case of $\neg \texttt{ROB}(s)$, which implies that $\tickets{C} \ge n-t$ and, by \ref{depositCoins:confiscated}, $\forall i\in [N] :\: u_i(s) = z_i(s)+y_i(s)$.
Consider the set of players $Q$ who are not better off after $C$'s deviation, $Q \triangleq \{ i\in[N] :\: u_i(s) \le u_i(s^*) \}$.

We show that $\tickets{Q}\geq t+1$.
Assume by way of contradiction that $\tickets{Q} < t+1$. Denote $\overline{Q}=[N] \setminus Q$ and note the following: 
\begin{enumerate}
    \item For all $i\in \overline{Q}$, $z_i+y_i-a_i > 0$. This stems directly from the definition of $\overline{Q}$ and the payoffs in $s$ and $s^*$. 
    \item $\sum_{i\in \overline{Q}} (z_i+y_i) \leq \sum_{i\in [N]} (z_i+y_i) < P$. Here, the first inequality is due to $s\in S_+$, and the second is by \ref{consumerCoins:sum_less_P} and \ref{externalCoins:sum-non-positive}. 
    \item $\tickets{\overline{Q}} \geq n-t = P$, by the contradiction assumption. 
\end{enumerate}
Hence,
\[ 0 < \sum_{i\in \overline{Q}} (z_i + y_i - a_i) = \sum_{i\in \overline{Q}} (z_i + y_i) - \tickets{\overline{Q}} < P-P = 0, \]
a contradiction. We conclude that $Q \cap C \neq \emptyset$ which implies that at least one member in $C$ is not better off after the deviation.
\end{proof}

Note that in this proof we have not used the restriction that players have a limited budget of external coins. \newline

%% file: proofs/Lemma_2.tex
\begin{lemmaclone} {Lemma:informing_deviation}
For all $C\subseteq [N]$ and for all $s_C \in \prod_{i\in C} S_i$, such that (i) \texttt{INF}$(s_{C},s^*_{-C})$ and (ii) $(s_{C},s^*_{-C}) \in S_+$, 
there exists a player $i\in C$ such that $u_i(s^*) \ge u_i(s_{C},s^*_{-C})$.
\end{lemmaclone}

\begin{proof}
Consider a coalition $C$ that deviates from $s^*$ resulting in a strategy vector $s \triangleq (s_{C},s^*_{-C})$ as assumed.
Denote the informant by $f \in [N]$.
Denote $\hat{C} \triangleq C \cup \{ f\}$ and observe that $\tickets{\hat{C}} \ge t+1$. ($f$ does not necessarily belong to $C$ as players who follow the default strategy inform in stage 2 if they can.) \texttt{INF}$(s)$ implies that $\forall i\ne f : u_i(s)=z_i(s)+y_i(s)$ by \ref{depositCoins:informing}.

Consider the set of players who are not better off after $C$'s deviation, $Q \triangleq \{ i\in[N] :\: u_i(s) \le u_i(s^*) \}$. Note that $f \notin Q$ as she makes a net profit from the deviation: $u_f(s) = \ell + z_f(s) + y_f(s)\ge \ell - e_f$ and $u_f(s^*)=a_f$. From \Cref{eqn:high-participation} we have $e_f+a_f \le \nicefrac{n}{3} < n =\ell$.

We show that $\tickets{Q}\geq n-t$.
Assume by way of contradiction that $\tickets{Q} < n-t$. Denote $\overline{Q}=[N] \setminus (Q \cup \{f\})$ and note the following: 
\begin{enumerate}
    \item For all $i\in \overline{Q}$, $z_i + y_i - a_i > 0$.  This stems directly from the definition of $\overline{Q}$.
    \item $\sum_{i\in \overline{Q}} (z_i + y_i) \leq \sum_{i\in [N] \setminus \{f\} } (z_i + y_i) < e_f + P$. Here, the first inequality is due to the fact that $s\in S_+$, and the second inequality is due to \ref{consumerCoins:sum_less_P}, \ref{externalCoins:sum-non-positive} and the fact that $y_f \ge -e_f$ (\ref{externalCoins:bounded_by_own_EC}).
    \item $\tickets{\overline{Q}} = (n-(a_f+\tickets{Q}) > n-a_f-(n-t) = t-a_f$. The inequality here is due to the contradiction assumption.
\end{enumerate}
Hence,
\begin{equation*}
\begin{split}
    0 &< \sum_{i\in \overline{Q}} (z_i + y_i - a_i) = \sum_{i\in \overline{Q}} (z_i + y_i) - \tickets{\overline{Q}} \\ &< (e_f + P) - (t - a_f) \le (2t-n)+(n-t)-t = 0,
\end{split}
\end{equation*}
a contradiction. In the last inequality we used \Cref{eqn:high-participation}, which implies that $e_f+a_f \le \nicefrac{n}{3}= 2t-n$. We conclude that $Q \cap \hat{C} \neq \emptyset$, and since $f \notin Q$, we conclude that at least one member in $C$ is not better off after the deviation.
\end{proof}

%% file: proofs/Lemma_3.tex
\begin{lemmaclone} {Lemma:collusion_CPNE}
For all $C\subseteq [N]$ and all $s_C \in \prod_{i\in C} S_i$, such that (i) $\neg$\texttt{SEC}$(s_{C},s^*_{-C})$ and (ii) $(s_{C},s^*_{-C}) \in S_+$, either $(s_{C},s^*_{-C})$ is not self-enforcing, or
there exists a player $i\in C$ such that $u_i(s^*) \ge u_i(s_{C},s^*_{-C})$.
\end{lemmaclone}

\begin{proof}
Consider a coalition $C$ that deviates from $s^*$ resulting in a strategy vector $s \triangleq (s_{C},s^*_{-C})$ as assumed.

If \texttt{INF}$(s)$, then from \Cref{Lemma:informing_deviation} there is (at least) one member of $C$ who is not better off after the deviation, and we are done. 

Otherwise, $\neg \texttt{INF}(s)$. Denote by $f$ the player that owns an account with the lowest id among those who obtain $\ge t+1$ shares after stage 1 (since secrecy breaks in $s$ there is at least one such player). Recall that in our model, $f$'s informing transaction (if issued) would precede all other informing transactions (if any). $f\in C$ as she does not inform, although the default strategy instructs her to do so. We next argue that there is a  member of $C$ that is better off deviating again, rendering $s$ not self-enforcing.
Consider three cases.
\begin{enumerate}
    \item Assume $u_f(s) \le a_f + P$. By \Cref{eqn:high-participation}, $e_f+a_f \le \nicefrac{n}{3} < \frac{2n}{3} = \ell - P $ which implies $\ell - e_f > a_f + P$. $\ell-e_f$ is $f$'s minimal payoff if she unilaterally deviates from $s$ and informs, and it is higher than $f$'s payoff under $s$. So, in this case, $s$ is not self-enforcing. 
    
    \item Assume $u_f(s) > a_f + P$ and \texttt{ROB}$(s)$. Since $\sum_{i=1}^N u_i(s) < P + \sum_{i=1}^N a_i$ (by \ref{payoff:sum_less_than}), the assumption implies that some player $j\in [N]$ ends up with $u_j(s) < a_j$.
    Note that if $j$ now deviates from $s$ and plays the default strategy (and therefore does not engage in side contracts), robustness continues to hold and she gains at least $a_j$ (by \ref{externalCoins:no_side_contract} and \ref{depositCoins:no_side_contract}). So in this case, too, $s$ is not self-enforcing.

    
    \item Assume $u_f(s) > a_f + P$ and $\neg \texttt{ROB}(s)$. Now, for all $i \in [N]: u_i(s) = z_i(s) + y_i(s)$ by \ref{depositCoins:confiscated}. So, by \ref{externalCoins:sum-non-positive} and \ref{consumerCoins:sum_less_P} $\sum_{i=1}^N u_i(s) < P$, which implies that some player $j\in [N]$ ends up with $u_j(s) < 0$. So, $s \notin S_+$ in contradiction to our assumption.\qedhere
\end{enumerate}
\end{proof}